\documentclass[english]{llncs}
\pdfoutput=1
\usepackage[T1]{fontenc}
\usepackage{color}
\definecolor{note_fontcolor}{rgb}{1, 0, 0}
\usepackage{babel}
\usepackage{float}
\usepackage{amsmath}
\usepackage{amssymb}
\usepackage{xargs}[2008/03/08]
\usepackage[unicode=true,
 bookmarks=true,bookmarksnumbered=false,bookmarksopen=false,
 breaklinks=false,pdfborder={0 0 0},backref=false,colorlinks=true]
 {hyperref}
\hypersetup{pdftitle={Decision Problems for Additive Regular Functions},
 pdfauthor={Rajeev Alur and Mukund Raghothaman}}

\makeatletter

\newenvironment{lyxgreyedout}
  {\textcolor{note_fontcolor}\bgroup\ignorespaces}
  {\ignorespacesafterend\egroup}
\floatstyle{ruled}
\newfloat{algorithm}{tbp}{loa}
\providecommand{\algorithmname}{Algorithm}
\floatname{algorithm}{\protect\algorithmname}

\numberwithin{equation}{section}
\numberwithin{figure}{section}

\usepackage{amssymb}
\usepackage{bussproofs}
\usepackage{centernot}
\usepackage{datetime}
\usepackage{filecontents}
\usepackage{microtype}
\usepackage{tikz}
\usepackage{pgfplots}
\usepackage{stmaryrd}

\usetikzlibrary{automata}
\usetikzlibrary{backgrounds}
\usetikzlibrary{fit}
\usetikzlibrary{positioning}

\@ifundefined{showcaptionsetup}{}{%
 \PassOptionsToPackage{caption=false}{subfig}}
\usepackage{subfig}
\makeatother

\begin{document}

\title{Decision~Problems~for~Additive~Regular~Functions}

\author{Rajeev Alur \and Mukund Raghothaman}

\institute{University of Pennsylvania\\
\email{\{\href{mailto:alur@cis.upenn.edu}{alur}, \href{mailto:rmukund@cis.upenn.edu}{rmukund}\}@cis.upenn.edu}}

\maketitle

\global\long\def\autobox#1{#1}

\global\long\def\opname#1{\autobox{\operatorname{#1}}}

\global\long\def\dontcare{\_}

\global\long\def\bbracket#1{\left\llbracket #1\right\rrbracket }

\global\long\def\cnot#1{\centernot#1}

\global\long\def\proofsep{\mathbin{\vdash}}

\global\long\def\fCenter{\proofsep}

\global\long\def\binaryprimitive#1#2{\BinaryInf$#1\fCenter#2$}

\global\long\def\axiomprimitive#1#2{\Axiom$#1\fCenter#2$}

\global\long\def\unaryprimitive#1#2{\UnaryInf$#1\fCenter#2$}

\global\long\def\trinaryprimitive#1#2{\TrinaryInf$#1\fCenter#2$}

\global\long\def\bussproof#1{#1\DisplayProof}

\global\long\def\binaryinfc#1#2#3#4{#1#2\RightLabel{\ensuremath{{\scriptstyle \textrm{#4}}}}\BinaryInfC{\ensuremath{#3}}}

\global\long\def\trinaryinfc#1#2#3#4#5{#1#2#3\RightLabel{\ensuremath{{\scriptstyle \textrm{#5}}}}\TrinaryInfC{\ensuremath{#4}}}

\global\long\def\unaryinfc#1#2#3{#1\RightLabel{\ensuremath{{\scriptstyle \textrm{#3}}}}\UnaryInfC{\ensuremath{#2}}}

\global\long\def\axiomc#1{\AxiomC{\ensuremath{#1}}}

\global\long\def\binaryinf#1#2#3#4#5{#1#2\RightLabel{\ensuremath{{\scriptstyle \textrm{#5}}}}\binaryprimitive{#3}{#4}}

\global\long\def\trinaryinf#1#2#3#4#5#6{#1#2#3\RightLabel{\ensuremath{{\scriptstyle \textrm{#6}}}}\trinaryprimitive{#4}{#5}}

\global\long\def\unaryinf#1#2#3#4{#1\RightLabel{\ensuremath{{\scriptstyle \textrm{#4}}}}\unaryprimitive{#2}{#3}}

\global\long\def\axiom#1#2{\axiomprimitive{#1}{#2}}

\global\long\def\axrulesp#1#2#3#4{\bussproof{\unaryinf{\axiomc{\vphantom{#4}}}{#1}{#2}{#3}}}

\global\long\def\axrule#1#2#3{\axrulesp{#1}{#2}{#3}{\proofsep\Gamma,\Delta}}

\global\long\def\unrule#1#2#3#4#5{\bussproof{\unaryinf{\axiom{#1}{#2}}{#3}{#4}{#5}}}

\global\long\def\unrulec#1#2#3{\bussproof{\unaryinfc{\axiomc{#1}}{#2}{#3}}}

\global\long\def\binrule#1#2#3#4#5#6#7{\bussproof{\binaryinf{\axiom{#1}{#2}}{\axiom{#3}{#4}}{#5}{#6}{#7}}}

\global\long\def\binrulec#1#2#3#4{\bussproof{\binaryinfc{\axiomc{#1}}{\axiomc{#2}}{#3}{#4}}}

\global\long\def\roset#1{\left\{  #1\right\}  }

\global\long\def\ruset#1#2{\roset{#1\;\middle\vert\;#2}}

\global\long\def\union#1#2{#1\cup#2}

\global\long\def\bigunion#1#2{\bigcup_{#1}#2}

\global\long\def\intersection#1#2{#1\cap#2}

\global\long\def\bigintersection#1#2{\bigcap_{#1}#2}

\global\long\def\bigland#1#2{\bigwedge_{#1}#2}

\global\long\def\powerset#1{2^{#1}}

\global\long\def\cart#1#2{#1\times#2}

\global\long\def\tuple#1{\left(#1\right)}

\global\long\def\quoset#1#2{\left.#1\middle/#2\right.}

\global\long\def\equivclass#1{\left[#1\right]}

\global\long\def\refltransclosure#1{#1^{*}}

\newcommandx\funcapphidden[3][usedefault, addprefix=\global, 1=]{#2#1#3}

\global\long\def\funcapplambda#1#2{\funcapphidden[\,]{#1}{#2}}

\global\long\def\funcapptrad#1#2{\funcapphidden{#1}{\tuple{#2}}}

\global\long\def\funccomp#1#2{#1\circ#2}

\global\long\def\arrow#1#2{#1\to#2}

\global\long\def\func#1#2#3{#1:\arrow{#2}{#3}}

\global\long\def\N{\mathbb{N}}

\global\long\def\Z{\mathbb{Z}}

\global\long\def\Q{\mathbb{Q}}

\global\long\def\R{\mathbb{R}}

\global\long\def\D{\mathbb{D}}

\global\long\def\Zmod#1{\quoset{\Z}{#1\Z}}

\global\long\def\vector#1{\mathbf{#1}}

\global\long\def\dotprod#1#2{#1\cdot#2}

\global\long\def\bool{{\tt bool}}

\global\long\def\true{{\tt true}}

\global\long\def\false{{\tt false}}

\global\long\def\inlinemod#1#2{#1\bmod{#2}}

\global\long\def\parenmod#1#2{#1\pmod{#2}}

\global\long\def\strempty{\epsilon}

\newcommandx\strcat[3][usedefault, addprefix=\global, 1=]{#2#1#3}

\global\long\def\strrev#1{#1^{\autobox{rev}}}

\global\long\def\strlen#1{\left|#1\right|}

\global\long\def\strlenp#1#2{\strlen{#1}_{#2}}

\global\long\def\prefix#1#2{\funcapptrad{\texttt{\ensuremath{\opname{pre}}}}{#1,#2}}

\global\long\def\suffix#1#2{\funcapptrad{\texttt{suf}}{#1,#2}}

\global\long\def\regexor#1#2{#1+#2}

\newcommandx\regexconcat[3][usedefault, addprefix=\global, 1=]{\strcat[#1]{#2}{#3}}

\global\long\def\kstar#1{#1^{*}}

\global\long\def\bigoh#1{\funcapptrad O{#1}}

\global\long\def\littleoh#1{\funcapptrad o{#1}}

\global\long\def\bigomega#1{\funcapptrad{\Omega}{#1}}

\global\long\def\littleomega#1{\funcapptrad{\omega}{#1}}

\global\long\def\bigtheta#1{\funcapptrad{\Theta}{#1}}

\global\long\def\logspace{\textsc{logspace}}

\global\long\def\nlogspace{\textsc{nlogspace}}

\global\long\def\np{\textsc{np}}

\global\long\def\conp{\textsc{co}\np}

\global\long\def\pspace{\textsc{pspace}}

\global\long\def\npspace{\textsc{npspace}}

\global\long\def\exptime{\textsc{exptime}}

\global\long\def\complete{\mbox{-complete}}

\global\long\def\hard{\mbox{-hard}}

\global\long\def\npc{\np\complete}

\global\long\def\nph{\np\hard}

\global\long\def\conpc{\conp\complete}

\global\long\def\pspacec{\pspace\complete}

\global\long\def\pspaceh{\pspace\hard}

\global\long\def\exptimec{\exptime\complete}

\global\long\def\exptimeh{\exptime\hard}

\global\long\def\crasimp{\operatorname{\mbox{CRA}}}

\global\long\def\cra#1{\funcapptrad{\crasimp}{#1}}

\global\long\def\ccrasimp{\operatorname{\mbox{CCRA}}}

\global\long\def\ccra#1{\funcapptrad{\ccrasimp}{#1}}

\global\long\def\sst{\operatorname{\mbox{SST}}}

\global\long\def\acra{\operatorname{\mbox{ACRA}}}

\global\long\def\dacra#1{\funcapptrad{\acra}{#1}}

\global\long\def\wa{\operatorname{\mbox{WA}}}

\newcommandx\oracleeval[1][usedefault, addprefix=\global, 1=\cdot]{\operatorname{\textsc{EVAL}}\left(#1\right)}

\newcommandx\oracleequiv[1][usedefault, addprefix=\global, 1=\cdot]{\operatorname{\textsc{EQUIV}}\left(#1\right)}

\global\long\def\bud{k}

\global\long\def\strat{\theta}

\begin{abstract}
Additive Cost Register Automata ($\acra$) map strings to integers
using a finite set of registers that are updated using assignments
of the form ``$x:=y+c$'' at every step. The corresponding class
of \emph{additive regular functions} has multiple equivalent characterizations,
appealing closure properties, and a decidable equivalence problem.
In this paper, we solve two decision problems for this model. First,
we define the \emph{register complexity} of an additive regular function
to be the minimum number of registers that an $\acra$ needs to compute
it. We characterize the register complexity by a necessary and sufficient
condition regarding the largest subset of registers whose values can
be made far apart from one another. We then use this condition to
design a $\pspace$ algorithm to compute the register complexity of
a given $\acra$, and establish a matching lower bound. Our results
also lead to a machine-independent characterization of the register
complexity of additive regular functions. Second, we consider \emph{two-player
games over $\acra$s}, where the objective of one of the players is
to reach a target set while minimizing the cost. We show the corresponding
decision problem to be $\exptimec$ when costs are non-negative integers,
but undecidable when costs are integers.
\end{abstract}
\begin{lyxgreyedout}

\newcommand{\branchshortorfull}[1]{}

\branchshortorfull{Neither branch enabled.}%
\end{lyxgreyedout}

\section{\label{sec:Introduction} Introduction}

Consider the following scenario: a customer frequents a coffee shop,
and each time purchases a cup of coffee costing $\$2$. At any time,
he may fill a survey, for which the store offers to give him a discount
of $\$1$ for each of his purchases that month (including for purchases
already made). We model this by the machine $M_{1}$ shown in figure
\ref{fig:Intro:M1:Coffee}. There are two states $q_{S}$ and $q_{\lnot S}$,
indicating whether the customer has filled out the survey during the
current month. There are three events to which the machine responds:
$C$ indicates the purchase of a cup of coffee, $S$ indicates completion
of the survey, and $\#$ indicates the end of a month. The registers
$x$, $y$ track how much money the customer owes the establishment:
in state $q_{\lnot S}$, the amount in $x$ assumes that he will not
fill out a survey that month, and the amount in $y$ assumes that
he will fill out a survey before the end of the month. At any time
the customer wishes to settle his account, the machine outputs the
amount of money owed, which is always the value in register $x$.

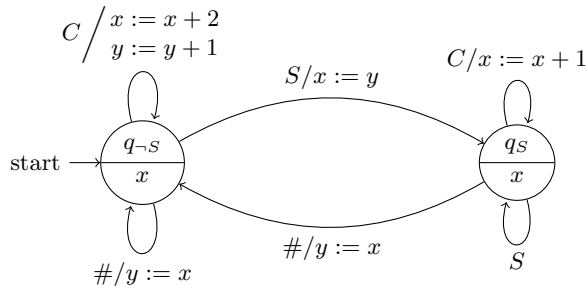
\begin{figure}
\begin{centering}
\pgfmathwidth{"$\quoset{C}{x := x + 2}$"}
\def\stringswidth{\pgfmathresult}
\pgfmathparse{2.0 * \stringswidth}
\global\edef\stringswidth{\pgfmathresult}

\begin{tikzpicture}[node distance=\stringswidth pt]

  \node [state with output, initial] (qnS) {$q_{\lnot S}$ \nodepart{lower} $x$};
  \node [state with output, right=of qnS] (qS) {$q_{S}$ \nodepart{lower} $x$};

  \path [->] (qnS) edge [loop above] node {$\quoset C {\begin{array}{c} x := x + 2\\ y := y + 1\end{array}}$} (qnS);
  \path [->] (qnS) edge [bend left, above] node {$\quoset S {x := y}$} (qS);
  \path [->] (qnS) edge [loop below] node {$\quoset \# {y := x}$} (qnS);

  \path [->] (qS) edge [loop above] node {$\quoset C {x := x + 1}$} (qS);
  \path [->] (qS) edge [loop below] node {$S$} (qS);
  \path [->] (qS) edge [bend left, below] node {$\quoset \# {y := x}$} (qnS);

\end{tikzpicture}
\par\end{centering}

\caption{\label{fig:Intro:M1:Coffee} $\acra$ $M_{1}$ models a customer in
a coffee shop. It implements a function $\func{f_{1}}{\kstar{\roset{C,S,\#}}}{\Z}$
mapping the purchase history of the customer to the amount he owes
the store.}
\end{figure}

The automaton $M_{1}$ has a finite state space, and a finite set
of integer-valued registers. On each transition, each register (say
$u$) is updated by an expression of the form ``$u:=v+c$'', for
some register $v$ and constant $c\in\Z$. Which of these registers
will eventually contribute to the output is determined by future events,
and so the cost of an event depends not only on the past, but also
on the future. Indeed, it can be shown that these machines are \emph{closed
under regular lookahead}, i.e. the register updates can be conditioned
on regular properties of an as-yet-unseen suffix, for no gain in expressivity.
The important limitation is that register updates are test-free, and
cannot examine the register contents.

The motivation behind the model is generalizing the idea of regular
languages to quantitative properties of strings. A language $L\subseteq\kstar{\Sigma}$
is regular when there is an accepting DFA. Regular languages are a
robust class, permitting multiple equivalent representations as regular
expressions and as formulas in monadic second-order logic. Recently
in \cite{CRA-Arxiv}, we proposed the model of regular functions:
they are the MSO-definable transductions from strings to expression
trees over some pre-defined grammar. The class of functions thus defined
depends on the grammar allowed; the simplest is when the underlying
domain is the set of integers $\Z$, and expressions involve constants
and binary addition, and we call these additive regular functions.
Additive regular functions have appealing closure properties, such
as under linear combination, input reversal, and regular lookahead,
and several analysis problems are efficiently decidable -- such as
containment, shortest paths and equivalence checking. $\acra$s correspond
to this class of additive regular functions.

Observe that machine $M_{1}$ has two registers, and it is not immediately
clear how (if it is even possible) to reduce this number. This is
the first question that this paper settles: Given an Additive Cost
Register Automaton ($\acra$) $M$, how do we determine the minimum
number of registers needed by any $\acra$ to compute $\bbracket M$?
We describe a phenomenon called register separation, and show that
any equivalent $\acra$ needs at least $k$ registers iff the registers
of $M$ are $k$-separable. It turns out that the registers of $M_{1}$
are $2$-separable, and hence two registers are necessary. We then
go on to show that determining $k$-separability is $\pspacec$. Determining
the register complexity is the natural analogue of the state minimization
problem for DFAs \cite{Hopcroft-Automata}.

The techniques used to analyse register complexity allow us to state
a result similar to the pumping lemma for regular languages: The register
complexity of $f$ is at least $k$ iff for some $m$, we have strings
$\sigma_{0}$, \ldots{}, $\sigma_{m}$, $\tau_{1}$, \ldots{}, $\tau_{m}$,
suffixes $w_{1}$, \ldots{}, $w_{k}$, and $k$ distinct coefficient
vectors $\mathbf{c}_{1},\ldots,\mathbf{c}_{k}\in\Z^{m}$ so that for
all vectors $\mathbf{x}\in\N^{m}$, $\funcapptrad f{\strcat{\strcat{\sigma_{0}}{\strcat{\tau_{1}^{x_{1}}}{\strcat{\sigma_{1}}{\strcat[\ldots]{\tau_{2}^{x_{2}}}{\sigma_{m}}}}}}{w_{i}}}=\sum_{j}c_{ij}x_{j}+d_{i}$.
Thus, depending on the suffix $w_{i}$, at least one of the cycles
$\tau_{1}$, \ldots{}, $\tau_{k}$ contributes differently to the
final cost.

Next, we consider $\acra$s with turn-based alternation. These are
games where several objective functions are simultaneously computed,
but only one of these objectives will eventually contribute to the
output, based on the actions of both the system and its environment.
Alternating $\acra$s are thus related to multi-objective games and
Pareto optimization \cite{MultiObjective}, but are a distinct model
because each run evaluates to a single value. We study the reachability
problem in $\acra$ games: Given a budget $\bud$, is there a strategy
for the system to reach an accepting state with cost at most $\bud$?
We show that this problem is $\exptimec$ when the registers assume
values from $\N$, and undecidable when the registers are integer-valued.

\subsubsection*{Related work}

\noindent The traditional model of string-to-number transducers has
been (non-deterministic) weighted automata (WA). Additive regular
functions are equivalent to unambiguous weighted automata over the
tropical semiring, and are therefore strictly sandwiched between weighted
automata and deterministic $\wa$s in expressiveness. Deterministic
$\wa$s are $\acra$s with one register, and algorithms exist to compute
the \emph{state complexity} and for minimization \cite{Mohri-EarliestNormalForm}.
Mohri \cite{Mohri-WA} presents a nice survey of the field. While
the determinizability of weighted automata remains an open problem
\cite{Kirsten-Twins,Kupferman-WA}, it has been solved in polynomial
time for the specific case of unambiguous weighted automata. There
is a polynomial translation from unambiguous $\wa$s to $\acra$s,
and the algorithm of subsection \ref{sub:RegMin:Pspace} runs in polynomial
time when the number of registers $k=2$. Thus, to the extent to which
they are relevant, we match the bounds available in the literature.
Recent work on the quantitative analysis of programs \cite{Chatterjee-QuantLang}
also uses weighted automata, but does not deal with minimization or
with notions of regularity. Data languages \cite{KaminskiFrancez-FiniteMemoryAutomata}
are concerned with strings over a (possibly infinite) data domain
$\D$. Recent models \cite{Bojanczyk-DataAutomata} have obtained
Myhill-Nerode characterizations, and hence minimization algorithms,
but the models are intended as acceptors, and not for computing more
general functions. Turn-based weighted games \cite{Markey-WeightedGames}
are $\acra$ games with a single register, and in this special setting,
it is possible to solve non-negative optimal reachability in polynomial
time. Of the techniques used in the paper, difference bound invariants
are a standard tool. However when we need them, in section \ref{sec:SepReg},
we have to deal with disjunctions of such constraints, and show termination
of invariant strengthening -- to the best of our knowledge, the relevant
problems have not been solved before.

\subsubsection*{Outline of the paper}

\noindent We define the automaton model in section \ref{sec:Prelim}.
In section \ref{sec:SepReg}, we introduce the notion of separability,
and establish its connection to register complexity. In section \ref{sec:RegMin},
we show that determining the register complexity is $\pspacec$. Finally,
in section \ref{sec:Games}, we study $\acra$ reachability games
-- in particular, that $\dacra{\Z}$ games are undecidable, and that
$\dacra{\N}$ reachability games are $\exptimec$.

\section{\label{sec:Prelim} Additive Regular Functions}

We will use additive cost register automata as the working definition
of additive regular functions, i.e. a function%
\footnote{By convention, we represent a partial function $\func fAB$ as a total
function $\func fA{B_{\bot}}$, where $B_{\bot}=\union B{\roset{\bot}}$,
and $\bot\notin B$ is the ``undefined'' value.%
} $\func f{\kstar{\Sigma}}{\Z_{\bot}}$ is regular iff it is implemented
by an $\acra$. An $\acra$ is a deterministic finite state machine,
supplemented by a finite number of integer-valued registers. Each
transition specifies, for each register $u$, a test-free update of
the form ``$u:=v+c$'', for some register $v$, and constant $c\in\Z$.
Accepting states are labelled with output expressions of the form
``$v+c$''.
\begin{definition}
\label{defn:ACRA}  An $\acra$ is a tuple $M=\tuple{Q,\Sigma,V,\delta,\mu,q_{0},F,\nu}$,
where $Q$ is a finite non-empty set of states, $\Sigma$ is a finite
input alphabet, $V$ is a finite set of registers, $\func{\delta}{\cart Q{\Sigma}}Q$
is the state transition function, $\func{\mu}{\cart Q{\cart{\Sigma}V}}{\cart V{\Z}}$
is the register update function, $q_{0}\in Q$ is the start state,
$F\subseteq Q$ is the non-empty set of accepting states, and $\func{\nu}F{\cart V{\Z}}$
is the output function.

The configuration of the machine is a pair $\gamma=\tuple{q,\autobox{val}}$,
where $q$ is the current state, and $\func{\autobox{val}}V{\Z}$
maps each register to its value. Define $\tuple{q,\autobox{val}}\opname{\to}^{a}\tuple{q^{\prime},\autobox{val}^{\prime}}$
iff $\funcapptrad{\delta}{q,a}=q^{\prime}$ and for each $u$, if
$\funcapptrad{\mu}{q,a,u}=\tuple{v,c}$, then $\funcapptrad{\autobox{val}^{\prime}}u=\funcapptrad{\autobox{val}}v+c$.

Machine $M$ then implements a function $\func{\bbracket M}{\kstar{\Sigma}}{\Z_{\bot}}$
defined as follows. For each $\sigma\in\kstar{\Sigma}$, let $\tuple{q_{0},\autobox{val}_{0}}\opname{\to}^{\sigma}\tuple{q_{f},\autobox{val}_{f}}$,
where $\funcapptrad{\autobox{val}_{0}}v=0$ for all $v$. If $q_{f}\in F$
and $\funcapptrad{\nu}{q_{f}}=\tuple{v,c}$, then $\funcapptrad{\bbracket M}{\sigma}=\funcapptrad{\autobox{val}_{f}}v+c$.
Otherwise $\funcapptrad{\bbracket M}{\sigma}=\bot$.
\end{definition}
We will write $\funcapptrad{\autobox{val}}{u,\sigma}$ for the value
of a register $u$ after the machine has processed the string $\sigma$
starting from the initial configuration. 

\begin{figure}[t]
\subfloat[\label{fig:Prelim:CRA:Examples:M2}$M_{2}$.]{\pgfmathwidth{"$\quoset{a}{x := x + 1}$"}
\def\stringswidth{\pgfmathresult}
\pgfmathparse{2.0 * \stringswidth}
\global\edef\stringswidth{\pgfmathresult}

\begin{tikzpicture}[above, sloped, node distance=\stringswidth pt]

  \node [state with output, initial] (q0) {$q_{0}$ \nodepart{lower} $x$};
  \node [state with output, right=of q0] (q1) {$q_{1}$ \nodepart{lower} $y$};

  \path [->] (q0) edge [loop above]
    node {$\quoset{a}{
      \begin{array}{l}
        x :=  x + 1\\
        y := y
      \end{array}}$}
    (q0);
  \path [->] (q0) edge [bend left]
    node {$\quoset{b}{
      \begin{array}{l}
        x := x\\
        y := y + 1
      \end{array}}$}
    (q1);

  \path [->] (q1) edge [loop above]
    node {$\quoset{b}{
      \begin{array}{l}
        x := x\\
        y := y + 1
      \end{array}}$}
    (q1);
  \path [->] (q1) edge [bend left]
    node [below] {$\quoset{a}{
      \begin{array}{l}
        x := x + 1\\
        y := y
      \end{array}}$}
    (q0);

\end{tikzpicture}

}\hfill{}\subfloat[\label{fig:Prelim:CRA:Examples:M3}$M_{3}$.]{\pgfmathwidth{"$\quoset {a}{x := 2}$"}
\def\stringswidth{\pgfmathresult}
\pgfmathparse{1.5 * \stringswidth}
\global\edef\stringswidth{\pgfmathresult}

\begin{tikzpicture}[above, sloped]

  \node [state with output, initial] (q0) {$q_{0}$ \nodepart{lower} $x$};

  \path [->] (q0) edge [loop above]
    node {$\quoset{a}{
      \begin{array}{l}
        x := y + 1\\
        y := y + 1\\
        z := z
      \end{array}}$}
    (q0);
  \path [->] (q0) edge [loop below]
    node {$\quoset{b}{
      \begin{array}{l}
        x := z + 1\\
        y := y\\
        z := z + 1
      \end{array}}$}
    (q0);

\end{tikzpicture}

}

\caption{\label{fig:Prelim:CRA:Examples} $\acra$s $M_{2}$ and $M_{3}$ operate
over the input alphabet $\Sigma=\roset{a,b}$. Both implement the
function defined as $\funcapptrad{f_{2}}{\strempty}=0$, and for all
$\sigma$, $\funcapptrad{f_{2}}{\protect\strcat{\sigma}a}=\strlen{\protect\strcat{\sigma}a}_{a}$,
and $\funcapptrad{f_{2}}{\protect\strcat{\sigma}b}=\strlen{\protect\strcat{\sigma}b}_{b}$.
Here $\strlen{\sigma}_{a}$ is the number of occurrences of the symbol
$a$ in the string $\sigma$.}
\end{figure}
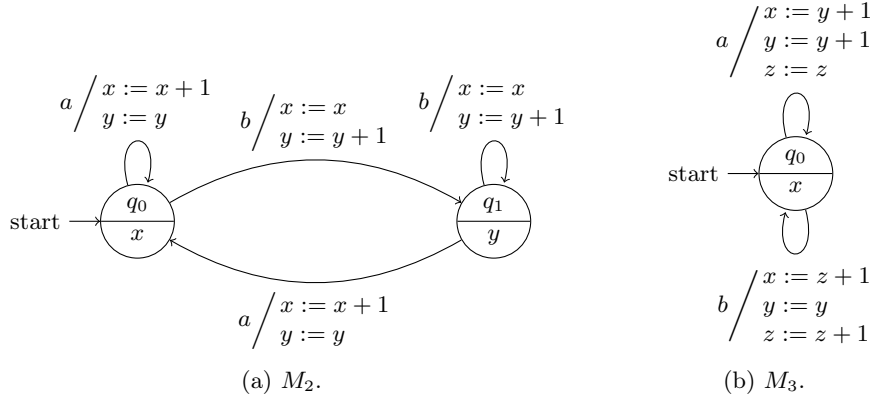

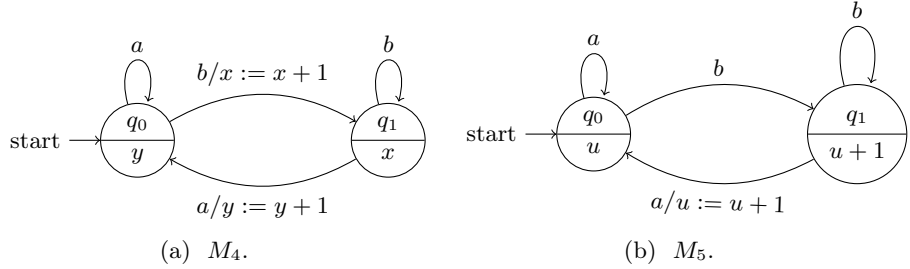
\begin{figure}[t]
\subfloat[\label{fig:Prelim:CRA:Examples:M4} $M_{4}$.]{\pgfmathwidth{"$\quoset{a}{x := x + 1}$"}
\def\stringswidth{\pgfmathresult}
\pgfmathparse{1.25 * \stringswidth}
\global\edef\stringswidth{\pgfmathresult}

\begin{tikzpicture}[above, sloped, node distance=\stringswidth pt]

  \node [state with output, initial] (q0) {$q_{0}$ \nodepart{lower} $y$};
  \node [state with output, right=of q0] (q1) {$q_{1}$ \nodepart{lower} $x$};

  \path [->] (q0) edge [loop above] node {$a$} (q0);
  \path [->] (q0) edge [bend left] node {$\quoset{b}{x := x + 1}$} (q1);
  \path [->] (q1) edge [loop above] node {$b$} (q1);
  \path [->] (q1) edge [bend left] node [below] {$\quoset{a}{y := y + 1}$} (q0);

\end{tikzpicture}

}\hfill{}\subfloat[\label{fig:Prelim:CRA:Examples:M5} $M_{5}$.]{\pgfmathwidth{"$\quoset{a}{x := x + 1}$"}
\def\stringswidth{\pgfmathresult}
\pgfmathparse{1.25 * \stringswidth}
\global\edef\stringswidth{\pgfmathresult}

\begin{tikzpicture}[above, sloped, node distance=\stringswidth pt]

  \node [state with output, initial] (q0) {$q_{0}$ \nodepart{lower} $u$};
  \node [state with output, right=of q0] (q1) {$q_{1}$ \nodepart{lower} $u + 1$};

  \path [->] (q0) edge [loop above] node {$a$} (q0);
  \path [->] (q0) edge [bend left] node {$b$} (q1);
  \path [->] (q1) edge [loop above] node {$b$} (q1);
  \path [->] (q1) edge [bend left] node [below] {$\quoset{a}{u := u + 1}$} (q0);

\end{tikzpicture}}

\caption{\label{fig:Prelim:CRA:Examples2} $\acra$s $M_{4}$ and $M_{5}$
operate over $\Sigma=\roset{a,b}$, and implement $f_{4}$ so that
if $\sigma$ ends in an $a$, then $\funcapptrad{f_{4}}{\sigma}=\mbox{number of }a\mbox{-s immediately following a }b$,
and otherwise $\funcapptrad{f_{4}}{\sigma}=\mbox{number of }b\mbox{-s immediately following an }a$.
When we omit the update for a register, say $v$, it is understood
to mean, ``$v:=v$''.}
\end{figure}

\begin{remark}
\label{rem:ACRA:Simpl:Trim}  Any given $\acra$ $M$ can easily
be \emph{trimmed} so that every state $q$ is reachable from the initial
state. All claims made in this paper assume that the machines under
consideration are trimmed.
\end{remark}
An important precondition when we define $k$-separability will be
that the registers be \emph{live}. Informally, a register $v$ is
live in state $q$ if for some suffix $\sigma\in\kstar{\Sigma}$,
on processing $\sigma$ starting $q$, the initial value of $v$ is
what influences the final output. For example, $M_{1}$ could be augmented
with a third register $z$ tracking the length of the string processed.
However, the value of $z$ would be irrelevant to the computation
of $f_{1}$, and $z$ would thus not be live. A straightforward way
of defining live registers is through \emph{suffix summaries}. Let
$q$ be a state, and $\sigma\in\kstar{\Sigma}$ be a string. Then
the suffix summary of $\sigma$ in $q$ is either a register-offset
pair $\cart V{\Z}$, or $\bot$, and which summarizes the effect of
processing $\sigma$ starting from state $q$. If the suffix summary
of $\sigma$ in $q$ is $\tuple{v,c}$, then it would be informally
read as: ``The result of processing suffix $\sigma$ if the machine
is currently in $q$ is the current value of $v$ plus $c$.'' Formally,
\begin{definition}
\label{defn:Suffix-Summary}  Let $q$ and $q^{\prime}$ be states
so that $\funcapptrad{\delta}{q,\sigma}=q^{\prime}$.
\begin{enumerate}
\item If $q^{\prime}\notin F$, then the suffix summary of $\sigma$ in
$q$ is $\bot$, and
\item (otherwise if $q^{\prime}\in F$) if $\funcapptrad{\nu}{q^{\prime}}=\tuple{u,c}$,
and $\funcapptrad{\mu}{q,\sigma,u}=\tuple{v,c^{\prime}}$, then the
suffix summary of $\sigma$ in $q$ is $\tuple{v,c+c^{\prime}}$.
\end{enumerate}
A register $v$ is \emph{live} in a state $q$ if for some $\sigma\in\kstar{\Sigma}$,
$c\in\Z$, the suffix summary of $\sigma$ in $q$ is $\tuple{v,c}$.\end{definition}
\begin{remark}
\label{rem:ACRA:Simpl:Live}  Whether a register $v$ is live in
a state $q$ is a static property of the state. At each state $q$,
pick a register $v_{q}$ which is live in $q$. If no such register
exists, then arbitrarily choose $v_{q}\in V$. On all transitions
into $q$, reset all non-live registers $v$ to the value of $v_{q}$.
This rewrite does not affect $\bbracket M$, and can be performed
in linear time. All claims made in this paper assume that this rewrite
has been performed.
\end{remark}
We recall the following properties of $\acra$s \cite{CRA-Arxiv}:

\subsubsection*{Equivalent characterizations}

Additive regular functions are equivalent to unambiguous weighted
automata \cite{Mohri-WA} over the tropical semiring. These are non-deterministic
machines with a single counter. Each transition increments the counter
by an integer $c$, and accepting states have output increments, also
integers. The unambiguous restriction requires that there be a single
accepting path for each string in the domain, thus the ``$\min$''
operation of the tropical semiring is unused. Consider the class of
MSO-definable string-to-integer transductions, with the successor
and predecessor operations allowed over integers. This class of functions
coincides with additive regular functions. Recently, streaming tree
transducers \cite{Loris-STT} have been proposed as the regular model
for string-to-tree transducers -- $\acra$s are equivalent in expressiveness
to regular string-to-term transducers with binary addition as the
base grammar.

\subsubsection*{Closure properties}

\noindent What makes additive%
\footnote{We will often drop the adjective ``additive'', and refer simply
to regular functions.%
} regular functions interesting to study is their robustness to various
manipulations:
\begin{enumerate}
\item for all $c\in\Z$, if $f_{1}$ and $f_{2}$ are regular functions,
then so are $f_{1}+f_{2}$ and $cf_{1}$,
\item if $f$ is a regular function, then $f_{rev}$ defined as $\funcapptrad{f_{rev}}{\sigma}=\funcapptrad f{\strrev{\sigma}}$
is also regular, and
\item if $f_{1}$ and $f_{2}$ are regular functions, and $L$ is a regular
language, then the function $f$ defined as $\funcapptrad f{\sigma}=\mbox{if }\sigma\in L,\mbox{ then }\funcapptrad{f_{1}}{\sigma},\mbox{ else }\funcapptrad{f_{2}}{\sigma}$
is also regular.
\item $\acra$s are closed under regular lookahead, i.e. even if the machine
were allowed to make decisions based on a regular property of the
suffix rather than simply the next input symbol, there would be no
increase in expressiveness.
\end{enumerate}

\subsubsection*{Analysis problems}

Given $\acra$s $M_{1}$ and $M_{2}$, equivalence-checking and the
min-cost problem ($\min_{\sigma\in\kstar{\Sigma}}\funcapptrad{\bbracket M}{\sigma}$)
can be solved in polynomial time. It follows then that containment
(for all $\sigma$, $\funcapptrad{\bbracket{M_{1}}}{\sigma}\leq\funcapptrad{\bbracket{M_{2}}}{\sigma}$)
also has a polynomial time algorithm.

\section{\label{sec:SepReg} Characterizing the Register Complexity}

The register complexity of a function $f$ is the minimum number of
registers an $\acra$ needs to compute it. For example the register
complexity of both $\bbracket{M_{1}}$ in figure \ref{fig:Intro:M1:Coffee}
and $\bbracket{M_{2}}$ in figure \ref{fig:Prelim:CRA:Examples:M2}
is $2$, while the register complexity of $\bbracket{M_{4}}$ is $1$.
Computing the register complexity is the first problem we solve, and
will occupy us for this section and the next.
\begin{definition}
\label{defn:RegComplexity}  Let $\func f{\kstar{\Sigma}}{\Z_{\bot}}$
be a regular function. The register complexity of $f$ is the smallest
number $k$ so there is an $\acra$ $M$ implementing $f$ with only
$k$ registers.
\end{definition}
Informally, the registers of $M$ are separable in some state $q$
if their values can be pushed far apart. For example, consider the
registers $x$, $y$ of $M_{1}$ in state $q_{0}$. For any constant
$c$, there is a string $\sigma=C^{c}$ leading to $q_{0}$ so that
$\left|\funcapptrad{\autobox{val}}{x,\sigma}-\funcapptrad{\autobox{val}}{y,\sigma}\right|\geq c$.

\begin{definition}
\label{defn:SepReg}  Let $M=\tuple{Q,\Sigma,V,\delta,\mu,q_{0},\nu}$
be an $\acra$. The registers of $M$ are \emph{$k$-separable} if
there is some state $q$, and a collection $U\subseteq V$ so that
\begin{enumerate}
\item $\left|U\right|=k$, all registers $v\in U$ are live in $q$, and
\item for all $c\in\Z$, there is a string $\sigma$, so that $\funcapptrad{\delta}{q_{0},\sigma}=q$
and for all distinct $u,v\in U$, $\left|\funcapptrad{\autobox{val}}{u,\sigma}-\funcapptrad{\autobox{val}}{v,\sigma}\right|\geq c$.
\end{enumerate}
\end{definition}
The registers of a machine $M$ are not $k$-separable if at every
state $q$, and collection $U$ of $k$ live registers, there is a
constant $c$ so for all strings $\sigma$ to $q$, $\left|\funcapptrad{\autobox{val}}{u,\sigma}-\funcapptrad{\autobox{val}}{v,\sigma}\right|<c$,
for some distinct $u,v\in U$. Note that the specific registers which
are close may depend on $\sigma$. For example, in machine $M_{3}$
from figure \ref{fig:Prelim:CRA:Examples:M3}, if the last symbol
was $a$, then $x$ and $y$ will be close, while if the last symbol
was a $b$, then $x$ and $z$ are guaranteed to be equal.
\begin{theorem}
\label{thm:SepReg} Let $\func f{\kstar{\Sigma}}{\Z_{\bot}}$ be a
function defined by an $\acra$ $M$. Then the register complexity
of $f$ is at least $k$ iff the registers of $M$ are $k$-separable.
\end{theorem}

The two directions of the proof are presented separately in the following
subsections.

\subsection{\label{full:sub:SepReg:Sep} $k$-separability implies a lower bound
on the register complexity}

Consider machine $M_{1}$ from figure \ref{fig:Intro:M1:Coffee}.
Here $k=2$, and registers $x$, $y$ are separated in state $q_{\lnot S}$.
Let $\sigma_{1}=\strempty$, i.e. the empty string, and $\sigma_{2}=S$
-- these are suffixes which, when starting from $q_{\lnot S}$, ``extract''
the values currently in $x$, $y$.

Now suppose an equivalent counter-example machine $M^{\prime}$ is
proposed with only one register $v$. At each state $q^{\prime}$
of $M^{\prime}$, observe the ``effect'' of processing suffixes
$\sigma_{1}$, $\sigma_{2}$. Each of these can be summarized by an
expression of the form $v+c_{q^{\prime}i}$ for $i\in\roset{1,2}$,
the current value of register $v$, and $c_{q^{\prime}i}\in\Z$. Thus,
the outputs differ by no more than $\left|\left(v+c_{q^{\prime}1}\right)-\left(v+c_{q^{\prime}2}\right)\right|\leq\left|c_{q^{\prime}1}\right|+\left|c_{q^{\prime}2}\right|$.
Fix $n=\funcapptrad{\max_{q^{\prime}}}{\left|c_{q^{\prime}1}\right|+\left|c_{q^{\prime}2}\right|}$,
and observe that for all $\sigma$, $\left|\funcapptrad{\bbracket{M^{\prime}}}{\strcat{\sigma}{\sigma_{1}}}-\funcapptrad{\bbracket{M^{\prime}}}{\strcat{\sigma}{\sigma_{2}}}\right|\leq n$.
For $\sigma=C^{n+1}$, $\left|\funcapptrad{f_{1}}{\strcat{\sigma}{\sigma_{1}}}-\funcapptrad{f_{1}}{\strcat{\sigma}{\sigma_{2}}}\right|>n$,
so $M^{\prime}$ cannot be equivalent to $M_{1}$. In general, by
a straightforward application of the pigeon-hole principle, we conclude:
\begin{lemma}
\label{lem:SepReg:101} Let $M$ be an $\acra$ whose registers are
$k$-separable. Then the register complexity of the implemented function
$f$ is at least $k$.
\end{lemma}
\begin{proof}
Assume otherwise, so we have a machine $M^{\prime}$ with only $k-1$
registers and equivalent to $M$. Let $q$ be that state of $M$ where
separation is achieved. For each $v\in U$, there is a suffix $\sigma_{v}\in\kstar{\Sigma}$
and constant $c_{v}$ so that the suffix summary of $\sigma_{v}$
in $q$ is $\tuple{v,c_{v}}$.

For each state $q^{\prime}$ of the proposed counter-example machine
$M^{\prime}$, and each register $v\in U$ of $M$, record the suffix
summary of $\sigma_{v}$ in $q^{\prime}$ -- $\tuple{v_{q^{\prime}v}^{\prime},c_{q^{\prime}v}^{\prime}}$,
or $\bot$. Define $c_{p}$ as:

\begin{align*}
c_{p} & =\funcapptrad{\max}{\max_{v\in U}\left|c_{v}\right|,\max_{q^{\prime},v\in U}\left|c_{q^{\prime}v}^{\prime}\right|}.
\end{align*}
Consider the state of the machine $M^{\prime}$ after processing some
prefix $\sigma_{pre}$. For each suffix $\sigma_{v}$, there must
be a register $v^{\prime}$ so that $\left|\funcapptrad{\bbracket{M^{\prime}}}{\strcat{\sigma_{pre}}{\sigma_{v}}}-\funcapptrad{\autobox{val}}{v^{\prime},\sigma_{pre}}\right|\leq c_{p}$.
Since there are only $k-1$ registers in $M^{\prime}$ and $k$ suffixes
$\sigma_{i}$, it must either be the case that for some pair $u,v\in U$,
this condition holds offset from the same register $v^{\prime}$.

We assumed the condition: for each $c\in\N$, there is a path $\sigma$
to $q$ so that $\left|\funcapptrad{\autobox{val}}{u,\sigma}-\funcapptrad{\autobox{val}}{v,\sigma}\right|\geq c$
(for all distinct $u,v\in V$). Instantiate this condition with $c=1+4c_{p}$,
and let $\sigma_{pre}=\sigma$ be the witness prefix. Let $\sigma_{u}$,
$\sigma_{v}$ be the pair of suffixes for which the suffix summaries
in $q^{\prime}$ depend on the same $v^{\prime}$. Since $\left|\funcapptrad{\autobox{val}}{u,\sigma_{pre}}-\funcapptrad{\autobox{val}}{v,\sigma_{pre}}\right|\geq c$,
it follows that $\left|\funcapptrad f{\strcat{\sigma_{pre}}{\sigma_{u}}}-\funcapptrad f{\strcat{\sigma_{pre}}{\sigma_{v}}}\right|\geq c-2c_{p}\geq1+2c_{p}$.
However, from our closeness condition, it follows that $\left|\funcapptrad{\bbracket{M^{\prime}}}{\strcat{\sigma_{pre}}{\sigma_{u}}}-\funcapptrad{\bbracket{M^{\prime}}}{\strcat{\sigma_{pre}}{\sigma_{v}}}\right|\leq2c_{p}$,
leading to a contradiction.\end{proof}

\subsection{\label{full:sub:SepReg:NecSep} Non-separability permits register
elimination}

\subsubsection{Intuition \label{ssub:SepReg:NecSep:Intuition}}

Say we are given an $\acra$ $M$, and told that its registers are
not $k$-separable. This can be rewritten in the form of an invariant
at each state: for each state $q$, there is a constant $c_{q}$ so
for every collection $U\subseteq V$ with $\left|U\right|=k$, and
for every string $\sigma$ with $\funcapptrad{\delta}{q_{0},\sigma}=q$,
there must exist distinct $u,v\in U$ with $\left|\funcapptrad{\autobox{val}}{u,\sigma}-\funcapptrad{\autobox{val}}{v,\sigma}\right|<c$.
For example, with $3$ registers $x$, $y$, $z$, this invariant
would be $\exists c,\left|x-y\right|<c\lor\left|y-z\right|<c\lor\left|z-x\right|<c$.
Now, if we know that $\left|x-y\right|<c$, then it suffices to explicitly
maintain the value of only one register, and the (bounded) difference
can be stored in the state.

Consider machines $M_{4}$, $M_{5}$ in figure \ref{fig:Prelim:CRA:Examples2}.
While $M_{4}$ is the intuitive first solution to the problem of implementing
$f_{4}$, the difference between registers $x$, $y$ is always bounded.
In both states, the non-separability invariant states $\left|x-y\right|\leq1$,
or $-1\leq x-y\leq1$. We exploit this to construct $M_{5}$, which
uses just one register $u$.

Since we need to track these register differences during execution,
the invariants must be inductive: if $D_{q}$ and $D_{q^{\prime}}$
are the invariants at states $q$, $q^{\prime}$, and $q\to^{a}q^{\prime}$
is a transition in the machine, then it must be the case that $D_{q}\implies\funcapptrad{\autobox{\textsc{wp}}}{D_{q^{\prime}},q,a}$.
Here $\autobox{\textsc{wp}}$ refers to the standard notion of the
weakest precondition from program analysis: the invariant $D_{q^{\prime}}$
identifies a set of variable valuations. $\funcapptrad{\autobox{\textsc{wp}}}{D_{q^{\prime}},q,a}$
is exactly that set of variable valuations $\autobox{val}$ so that
$\tuple{q,\autobox{val}}\to^{a}\tuple{q^{\prime},\autobox{val}^{\prime}}$
for some $D_{q^{\prime}}$-satisfying valuation $\autobox{val}^{\prime}$.

The standard technique to make a collection of invariants inductive
is strengthening: if $D_{q}\cnot{\implies}\funcapptrad{\autobox{\textsc{wp}}}{D_{q^{\prime}},q,a}$,
then $D_{q}$ is replaced with $D_{q}\land\funcapptrad{\autobox{\textsc{wp}}}{D_{q^{\prime}},q,a}$,
and this process is repeated at every pair of states until fixpoint.
This procedure is seeded with the invariants asserting non-separability.
However, before the result of this back-propagation can be used in
our arguments, we must prove that the method terminates -- this is
the main technical problem solved in this section.

We now sketch a proof of this termination claim for a simpler class
of invariants. Consider the class of difference-bound constraints
-- assertions of the form $C=\bigland{u,v\in V}{a_{uv}<u-v<b_{uv}}$,
where for each $u$, $v$, $a_{uv},b_{uv}\in\Z$ or $a_{uv},b_{uv}\in\roset{-\infty,\infty}$.
Observe that $C$ induces an equivalence relation $\equiv_{C}$ over
the registers: $u\equiv_{C}v$ iff $a_{uv},b_{uv}\in\Z$. Let $C$
and $C^{\prime}$ be some pair of constraints so that $C\cnot{\implies}C^{\prime}$,
so that the assertion $C\land C^{\prime}$ is strictly stronger than
$C$. Either $C\land C^{\prime}$ relates a strictly larger set of
variables -- $\equiv_{C}\subsetneq\equiv_{C\land C^{\prime}}$ --
or (if $\equiv_{C}=\equiv_{C\land C^{\prime}}$) for some pair of
registers $u$, $v$, the bounds $a_{uv}^{\prime}<u-v<b_{uv}^{\prime}$
imposed by $C\land C^{\prime}$ are a strict subset of the bounds
$a_{uv}<u-v<b_{uv}$ imposed by $C$. Observe that the first type
of strengthening can happen at most $\left|V\right|^{2}$ times, while
the second type of strengthening can happen only after $a_{uv}$,
$b_{uv}$ are established for a pair of registers $u$, $v$, and
can then happen at most $b_{uv}-a_{uv}$ times. Thus the process of
repeated invariant strengthening must terminate. However, the statements
asserting non-separability are disjunctions of difference-bound constraints.
We show that the above insight is sufficient even for this generalization.

The rest of this subsection is devoted to formalizing the intuition
presented above.

\subsubsection{Difference bound constraints and well-formed invariants \label{ssub:SepReg:NecSep:Defn}}
\begin{definition}
\label{defn:SepReg:DiffBound} A \emph{difference bound constraint}
is a conjunction of constraints of the form $a<u-v<b$, for $a,b\in\union{\Z}{\roset{-\infty,\infty}}$
(and either $a$, $b$ are both finite, or both infinite), and $u,v\in V$.
\emph{Well-formed invariants} are finite disjunctions of difference
bound constraints.
\end{definition}
Note that if there is a non-trivial term corresponding to $u-v$ in
a difference bound constraint, then the difference is bounded both
from above and below, i.e. $a<u-v<b$, and $a,b\in\Z$. For example,
$0<u-v<\infty$ is not a difference bound constraint. The trivial
constraint $-\infty<u-v<\infty$ holds of every pair of registers.
Given a difference bound constraint $C$, it can be set in \emph{closed
form} where whenever $C$ contains the term $a<u-v<b$ it also contains
$-b<v-u<-a$, and if $C$ contains the terms $a<u-v<b$ and $a^{\prime}<v-w<b^{\prime}$,
then it also contains the term $a^{\prime\prime}<u-w<b^{\prime\prime}$,
for some $a+a^{\prime}\leq a^{\prime\prime}\leq b^{\prime\prime}\leq b+b^{\prime}$.
A difference bound constraint establishes an equivalence relation
over the registers of $V$, where $u\equiv v$ iff there is a constant
$c$ so that $C\implies\left|u-v\right|<c$. This is the same as saying
that $u\equiv v$ iff $C$ in closed form contains a non-trivial term
corresponding to $u-v$. The following proposition describes exactly
the cases when a difference-bound constraint $C$ is strictly stronger
than another constraint $C^{\prime}$:
\begin{claim}
\label{clm:SepReg:Strengthening} Let $C=c_{1}\land c_{2}\land\ldots\land c_{k}$
and $C^{\prime}=c_{1}^{\prime}\land c_{2}^{\prime}\land\ldots\land c_{k^{\prime}}^{\prime}$
be difference bound constraints. If $C$ is strictly stronger than
$C^{\prime}$, i.e. $C\implies C'$ but $C^{\prime}\cnot{\implies}C$,
then either
\begin{enumerate}
\item $\equiv^{\prime}\subsetneq\equiv$, where $\equiv$, $\equiv^{\prime}$
are the equivalence relations over $V$ generated by $C$ , $C^{\prime}$,
or
\item (otherwise if $\equiv^{\prime}=\equiv$) for some registers $u,v\in V$,
the best bounds $a<u-v<b$ and $a^{\prime}<u-v<b^{\prime}$ implied
by $C$ and $C^{\prime}$ are related as $\roset{a,a+1,a+2,\ldots,b}\subsetneq\roset{a^{\prime},a^{\prime}+1,a^{\prime}+2,\ldots,b^{\prime}}$.
\end{enumerate}
\end{claim}

\subsubsection{Well-formed invariants are well-ordered \label{ssub:SepReg:NecSep:Strengthening}}
\begin{lemma}
\label{lem:SepReg:NoInfChains} Let $T$ be a labeled tree, where
each node $u$ is labeled with a difference bound constraint $C_{u}$,
and is of finite degree. Say also that the constraint at each node
is strictly stronger than the constraint at its parent. Then $T$
cannot be infinite.\end{lemma}
\begin{proof}
Assume otherwise. By K\"{o}nig's lemma, there must be an infinite
path through this tree, and the constraints along this path strictly
increase in strength. We now argue that such a path cannot exist.

Observe that the equivalence relation $\equiv$ associated with a
difference bound constraint $C$ can have no more than $\left|V\right|^{2}$
elements. Also, once we have a pair of registers constrained as $a<u-v<b$,
(with both $a$, $b$ finite), the constraint can be tightened only
$b-a$ times. Furthermore, such tightening can only happen after $u\equiv v$,
by the equivalence relation $\equiv$ associated with $C$. Thus,
every sequence of difference bound constraints strictly increasing
in strength must be finite. This completes the proof.\end{proof}
\begin{definition}
\label{defn:SepReg:WP} Let $\funcapptrad{\varphi}{\autobox{val}}$
be an arbitrary formula that identifies sets of states. Let $q,q^{\prime}\in Q$
be two states so that $q^{\prime}\to^{a}q$ for some symbol $a\in\Sigma$.
Then, the weakest precondition of $\varphi$ at $q$ with respect
to the transition from $q^{\prime}$ on $a$, written as $\varphi^{\prime}=\funcapptrad{\autobox{\textsc{wp}}}{\varphi,q^{\prime},a}$
is $\funcapptrad{\varphi^{\prime}}{\autobox{val}^{\prime}}\iff\forall\autobox{val},\tuple{q^{\prime},\autobox{val}^{\prime}}\to^{a}\tuple{q,\autobox{val}}\implies\funcapptrad{\varphi}{\autobox{val}}$.
\end{definition}
It can be shown that $\funcapptrad{\autobox{\textsc{wp}}}{\varphi,q^{\prime},a}$
can be obtained by simultaneously replacing every occurrence of each
register with its update expression over the transition: $\varphi^{\prime}=\varphi\left[v\mapsto\funcapptrad{\mu}{q^{\prime},a,v}\right]_{v}$,
where the update expression $\funcapptrad{\mu}{q^{\prime},a,v}=\tuple{u,c}$
is read as ``$u+c$''. For example, consider machine $M_{4}$ in
figure \ref{fig:Prelim:CRA:Examples:M4}: the weakest precondition
of the assertion $-2<x-y<2$ in state $q_{1}$ with respect to the
transition on $b$ from $q_{0}$ is the assertion $-2<x+1-y<2$, or
$-3<x-y<1$. It can be shown that:
\begin{claim}

\begin{enumerate}
\item Let $D_{q^{\prime}}$ be a well-formed invariant in some state $q^{\prime}$
of an $\acra$ $M$. Let $q\in Q$ and $a\in\Sigma$ so $\funcapptrad{\delta}{q,a}=q^{\prime}$.
Then $\funcapptrad{\autobox{\textsc{wp}}}{D_{q^{\prime}},q,a}$ is
also a well-formed invariant.
\item Let $D$ and $D^{\prime}$ be well-formed invariants. Then so is $D\land D^{\prime}$.
\end{enumerate}
\end{claim}
\begin{algorithm}
\begin{enumerate}
\item \label{enu:SepReg:Saturate:TreeInit} At each state $q$, initialize
a tree $T_{q}$. Nodes of this tree are labeled with difference bound
constraints. The root of each tree $T_{q}$ is $true$, and its immediate
children are the difference bound constraints $C$ in $D_{q}$.
\item \label{enu:SepReg:Saturate:While} While there exist states $q,q^{\prime}\in Q$
and symbols $a\in\Sigma$, so that $\funcapptrad{\delta}{q,a}=q^{\prime}$,
but $D_{q}\cnot{\implies}\funcapptrad{\autobox{\textsc{wp}}}{D_{q^{\prime}},q,a}$.
For each difference bound constraint $C\in D_{q}$ so that $C\cnot{\implies}\funcapptrad{\autobox{\textsc{wp}}}{D_{q^{\prime}},q,a}$:

\begin{enumerate}
\item Calculate $C\land\funcapptrad{\autobox{\textsc{wp}}}{D_{q^{\prime}},q,a}$,
by the distributivity of the logical $\textsc{and}$ operator over
$\textsc{or}$.
\item \label{enu:SepReg:Saturate:TreeGrow} For the node corresponding to
$C$ in $T_{q}$, create children corresponding to each disjunct in
$C\land\funcapptrad{\autobox{\textsc{wp}}}{D_{q^{\prime}},q,a}$.
\item Replace $C$ at $D_{q}$ with the disjuncts in $C\land\funcapptrad{\autobox{\textsc{wp}}}{D_{q^{\prime}},q,a}$.
\end{enumerate}
\item Return, for each state $q$, the well-formed constraint $D_{q}$.
\end{enumerate}
\caption{\label{alg:SepReg:Saturate} $\textsc{saturate}$. Given an $\acra$
$M$, and a well-formed invariant $D_{q}$ at each state $q\in Q$.
The algorithm returns an inductive strengthening of these invariants.}
\end{algorithm}

\begin{lemma}
\label{lem:SepReg:Saturate:Term} For every input $\tuple{M,D_{q\in Q}}$,
algorithm \ref{alg:SepReg:Saturate} terminates.\end{lemma}
\begin{proof}
Observe that with each iteration of the loop in step \ref{enu:SepReg:Saturate:While},
the size of $T_{q}$ increases, for some $q$. If the algorithm were
to not terminate, then for some $q$, $T_{q}$ would be infinite.
We maintain the invariant that each node in $T_{q}$ has finite degree,
and the difference bound constraint at each node is strictly stronger
than that at its predecessor. But lemma \ref{lem:SepReg:NoInfChains}
tells us that no such infinite tree $T_{q}$ can exist.
\end{proof}

\subsubsection{Putting it all together: Constructing $M^{\prime}$ \label{ssub:SepReg:NecSep:Knockout}}
\begin{lemma}
\label{lem:SepReg:101C} Consider an $\acra$ $M$ whose registers
are not $k$-separable. Then, we can effectively construct an equivalent
machine $M^{\prime}$ with only $k-1$ registers.
\end{lemma}
\begin{proof}
The idea is that the difference bounds allow us to track all but $k-1$
registers in the state. So some registers $u$ are represented in
the state as a pair $\tuple{v,c}$, and we simulate the effect of
register $u$ by the expression $v+c$.

Since the registers of $M$ are not $k$-separable, at each state
$q$, and collection of $k$ registers $U$, there is a constant $c$
so for all paths $\sigma$ going to $q$, there is some pair of distinct
registers $u,v\in U$ so that $\left|\funcapptrad{\autobox{val}}{u,\sigma}-\funcapptrad{\autobox{val}}{v,\sigma}\right|<c$
(or equivalently, $-c<u-v<c$). Since $U\in\powerset V$ is drawn
from a finite set, and any instantiation of $c$ can be replaced by
a larger constant $c^{\prime}\geq c$, we can change the order of
quantifiers: at each state $q$, there is a constant $c$, so for
all paths $\sigma$ going to $q$ and collections of $k$ registers
$U\subseteq V$, there exist distinct $u,v\in U$ so that $\left|\funcapptrad{\autobox{val}}{u,\sigma}-\funcapptrad{\autobox{val}}{v,\sigma}\right|<c$.
Simplifying this, we obtain at each state $q$, a well-formed invariant
$D_{q}$. In each disjunct $C$ in $D_{q}$, there is never a collection
of more than $k-1$ mutually unrelated registers. Run $\autobox{\textsc{saturate}}$
on these constraints to make them inductive.

Now construct $M^{\prime}$ as follows. Consider some state $q$ and
some difference bound constraint $C\in D_{q}$. Now arbitrarily pick
a maximal set $V_{q,C}\subsetneq V$ of registers so no two elements
$u,v\in V_{q,C}$ are constrained by $C$. Since this set is maximal,
for every register $u\in V\setminus V_{q,C}$, there is a register
$v\in V_{q,C}$ so we have $C\implies a_{q,C,u}\leq u-v\leq b_{q,C,u}$,
for $a_{q,C,u},b_{q,C,u}\in\Z$. Define the state space $Q^{\prime}$
of $M^{\prime}$ as: 
\begin{align*}
Q^{\prime} & =\bigunion{q,C\in D_{q}}{\left(\cart{\roset{\tuple{q,C}}}{\prod_{u\in V\setminus V_{q,C}}\left[a_{q,C,u},b_{q,C,u}\right]}\right)},
\end{align*}
where $\left[a_{q,C,u},b_{q,C,u}\right]$ is the set of integers $a_{q,C,u}\leq z\leq b_{q,C,u}$.
Thus, for example, if we have $3$ registers $x$, $y$, $z$, and
at state $q$, we have the invariant that $-2\leq x-y\leq3$, and
$0\leq z\leq1$, then $q$ would produce states $\roset{\tuple{q,-2,0},\tuple{q,-2,1},\tuple{q,-1,0},\tuple{q,-1,1},\tuple{q,0,0},\tuple{q,0,1},\ldots,\tuple{q,3,1}}$.
Also, $V_{q,C}$ never has more than $k-1$ registers.

Now define $\func{\delta^{\prime}}{\cart{Q^{\prime}}{\Sigma}}{Q^{\prime}}$.
Let $\tuple{q,C,\mathbf{v}}\in Q^{\prime}$ be a state, where $\mathbf{v}$
refers to the values of the offsets. Let $a$ be a symbol, and let
$\funcapptrad{\delta}{q,a}=q^{\prime}$. Since the invariants are
inductive, it follows that there is a difference bound constraint
$C^{\prime}$ at $q^{\prime}$ which holds when the machine makes
this transition with this precondition. Also, there is enough information
to determine statically the values of the offsets $\mathbf{v}^{\prime}$.
Define $\funcapptrad{\delta^{\prime}}{\tuple{q,C,\mathbf{v}},a}=\tuple{q^{\prime},C^{\prime},\mathbf{v}^{\prime}}$.

Let $k^{\prime}=\max_{q,C}\left|V_{q,C}\right|$. Define $V^{\prime}$
to have $k^{\prime}$ registers. At each state-constraint pair $q$,
$C$, choose an arbitrary mapping scheme which maps registers $v^{\prime}\in V^{\prime}$
to registers $v\in V_{q,C}$. The invariant is that for all paths
to $\tuple{q,C,\mathbf{v}}$, $v^{\prime}$ holds the value of the
corresponding register $v$. For every register $u\in V\setminus V_{q,C}$,
the offsets in $\mathbf{v}$ provide enough information to simulate
its value by the expression $v+c$. Because the invariants are inductive,
there is enough local information to define the register update function
$\mu^{\prime}$, and the output function $\nu^{\prime}$.

The start state $q^{\prime}$ is any triple $\tuple{q_{0},C,\mathbf{0}}$,
where $C$ is any constraint at $q_{0}$ satisfied initially. All
registers start at $0$, so all register differences start at $0$
also. Observe that the machine $M^{\prime}$ is equivalent to $M$
by construction, and has $k^{\prime}<k$ registers. This completes
the proof.
\end{proof}
It should be noted that there is considerable freedom when defining
the reduced machine $M^{\prime}$ above: the start state $\tuple{q_{0},C,\vector 0}$
is not necessarily unique -- any difference-bound constraint $C\in D_{q_{0}}$
which is initially satisfied will work. Also, there may be multiple
difference-bound constraints $C_{1}^{\prime}$, $C_{2}^{\prime}$,
\ldots{}, that are satisfied at $q^{\prime}$ when making a transition
on symbol $a$ from $\tuple{q,C,\vector x}$. The choice in such cases
can be made arbitrarily.
\begin{example}
Consider machine $M_{3}$ in figure \ref{fig:Prelim:CRA:Examples:M3}.
By construction, we know that register $x$ always holds the same
value as one of the registers $y$, $z$. In particular, we have $\left|x-y\right|\leq0\lor\left|y-z\right|\leq0\lor\left|z-x\right|\leq0$
as the non-separation invariant. The weakest precondition with respect
to the transition from $q$ on $a$ is $\left|\left(y+1\right)-\left(y+1\right)\right|\leq0\lor\left|\left(y+1\right)-z\right|\leq0\lor\left|z-\left(y+1\right)\right|\leq0$,
which is always true. Thus, $D_{q}\implies\funcapptrad{\autobox{\textsc{wp}}}{D_{q},q,a}$,
and similarly $D_{q}\implies\funcapptrad{\autobox{\textsc{wp}}}{D_{q},q,b}$.
Algorithm \ref{alg:SepReg:Saturate} returns immediately. We then
construct the $3$ state machine shown in figure \ref{fig:SepReg:101C:Example}.
State $q_{xy}$ encodes the triple $\tuple{q_{0},x=y,0}$, and similarly
for $q_{yz}$ and $q_{zx}$. The machine maintains $2$ registers
$u$, $v$. The state-specific mapping of these to the original registers
are: in $q_{xy}$, $u$, $v$ hold $x$, $z$, in $q_{yz}$, $u$,
$v$ hold $x$, $y$, and in $q_{zx}$, $u$, $v$ hold $z$, $y$
respectively. Any of the states could be marked as the start state.

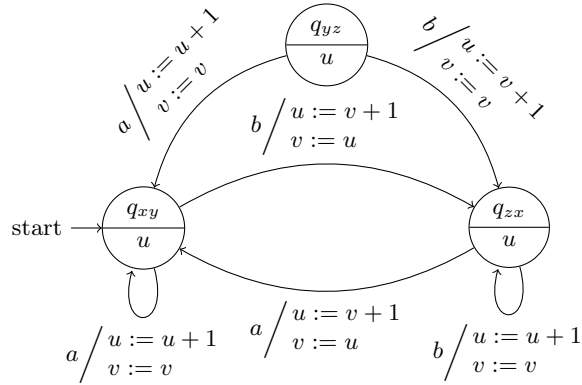
\begin{figure}
\begin{centering}
\pgfmathwidth{"$\quoset{a}{x := x + 1}$"}
\def\stringswidth{\pgfmathresult}
\pgfmathparse{1.25 * \stringswidth}
\global\edef\stringswidth{\pgfmathresult}

\begin{tikzpicture}[above, sloped, node distance=\stringswidth pt]

  \node [state with output, initial] (q0xy) {$q_{xy}$ \nodepart{lower} $u$};
  \node [state with output, above right=of q0xy] (q0yz) {$q_{yz}$ \nodepart{lower} $u$};
  \node [state with output, below right=of q0yz] (q0zx) {$q_{zx}$ \nodepart{lower} $u$};

  \path [->] (q0xy) edge [loop below] node {$\quoset{a}{\begin{array}{l} u := u + 1\\ v := v\end{array}}$} (q0xy);
  \path [->] (q0xy) edge [bend left] node {$\quoset{b}{\begin{array}{l} u := v + 1\\ v := u\end{array}}$} (q0zx);
  \path [->] (q0zx) edge [loop below] node {$\quoset{b}{\begin{array}{l} u := u + 1\\ v := v\end{array}}$} (q0zx);
  \path [->] (q0zx) edge [bend left, below] node {$\quoset{a}{\begin{array}{l} u := v + 1\\ v := u\end{array}}$} (q0xy);

  \path [->] (q0yz) edge [bend right] node {$\quoset{a}{\begin{array}{l} u := u + 1\\ v := v\end{array}}$} (q0xy);
  \path [->] (q0yz) edge [bend left] node {$\quoset{b}{\begin{array}{l} u := v + 1\\ v := v\end{array}}$} (q0zx);

\end{tikzpicture}
\par\end{centering}

\caption{\label{fig:SepReg:101C:Example} An example application of lemma \ref{lem:SepReg:101C}
to $M_{3}$.}

\end{figure}
\end{example}

\section{\label{sec:RegMin} Computing the Register Complexity}

\subsection{\label{sub:RegMin:Pspace} Computing the register complexity is in
$\pspace$}

\subsubsection{Intuition \label{ssub:RegMin:Pspace:Intuition}}

We reduce the problem of determining the register complexity of $\bbracket M$
to one of determining reachability in a directed ``register separation''
graph with $\bigoh{\left|Q\right|2^{\left|V\right|^{2}}}$ nodes.
The presence of an edge in this graph can be determined in polynomial
space, and thus we have a $\pspace$ algorithm to determine the register
complexity. Otherwise, if polynomial time algorithms are used for
graph reachability and $1$-counter $0$-reachability, the procedure
runs in time $\bigoh{c^{3}\left|Q\right|^{4}2^{4\left|V\right|^{2}}}$,
where $c$ is the largest constant in the machine.

We first generalize the idea of register separation to that of separation
relations: an arbitrary relation $\opname{\parallel}\subseteq\cart VV$
separates a state $q$ if for every $c\in\Z$, there is a string $\sigma$
so that $\funcapptrad{\delta}{q_{0},\sigma}=q$, and whenever $u\parallel v$,
$\left|\funcapptrad{\autobox{val}}{u,\sigma}-\funcapptrad{\autobox{val}}{v,\sigma}\right|\geq c$.
Thus, the registers of $M$ are $k$-separable iff for some state
$q$ and some subset $U$ of live registers at $q$, $\left|U\right|=k$
and $\ruset{\tuple{u,v}}{u,v\in U,u\neq v}$ separates $q$.

Consider a string $\tau\in\kstar{\Sigma}$, so for some $q$, $\funcapptrad{\delta}{q,\tau}=q$.
Assume also that:
\begin{enumerate}
\item For every register $u$ in the domain or range of $\parallel$, $\funcapptrad{\mu}{q,\tau,u}=\tuple{u,c_{u}}$,
for some $c_{u}\in\Z$, and
\item for some pair of registers $x$, $y$, $\funcapptrad{\mu}{q,\tau,x}=\tuple{x,c}$
and $\funcapptrad{\mu}{q,\tau,y}=\tuple{y,c^{\prime}}$ for distinct
$c$, $c^{\prime}$.
\end{enumerate}
Thus, every pair of registers that is already separated is preserved
during the cycle, and some new pair of registers is incremented differently.
We call such strings $\tau$ ``separation cycles'' at $q$. They
allow us to make conclusions of the form: If $\parallel$ separates
$q$, then $\union{\parallel}{\roset{\tuple{x,y}}}$ also separates
$q$.

Now consider a string $\sigma\in\kstar{\Sigma}$, so for some $q$,
$q^{\prime}$, $\funcapptrad{\delta}{q,\sigma}=q^{\prime}$. Pick
arbitrary relations $\parallel$, $\parallel^{\prime}$, and assume
that whenever $u^{\prime}\parallel^{\prime}v^{\prime}$, and $\funcapptrad{\mu}{q,\sigma,u^{\prime}}=\tuple{u,c_{u}}$,
$\funcapptrad{\mu}{q,\sigma,v^{\prime}}=\tuple{v,c_{v}}$, we have
$u\parallel v$. We can then conclude that if $\parallel$ separates
$q$, then $\parallel^{\prime}$ separates $q^{\prime}$ We call such
strings $\sigma$ ``renaming edges'' from $\tuple{q,\parallel}$
to $\tuple{q^{\prime},\parallel^{\prime}}$.

We then show that if $\parallel$ separates $q$ and $\parallel$
is non-empty, then there is a separation cycle-renaming edge sequence
to $\tuple{q,\parallel}$ from some strictly smaller separation $\tuple{q^{\prime},\parallel^{\prime}}$.
Thus, separation at each node can be demonstrated by a sequence of
separation cycles with renaming edges in between, and thus we reduce
the problem to that of determining reachability in an exponentially
large register separation graph. Finally, we show that each type of
edge can be determined in $\pspace$.

\subsubsection{Register separation graphs \label{ssub:RegMin:Pspace:SepRel}}

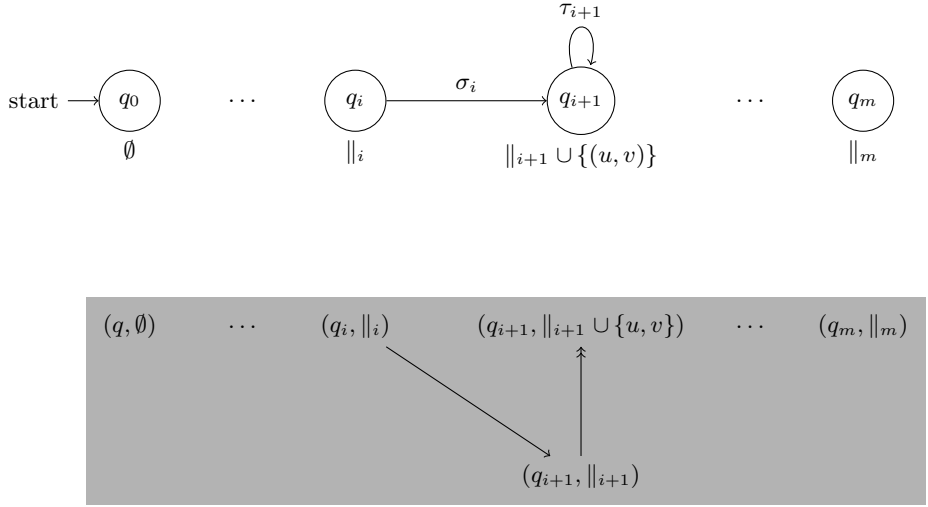
\begin{figure}
\begin{centering}
\begin{tikzpicture}[sloped]

  \node [state, initial, label=below:{$\emptyset$}] (q0) at (0, 0) {$q_{0}$};
  \node (qd1) at (1.5, 0) {$\ldots$};
  \node [state, label=below:{$\parallel_{i}$}] (qi) at (3, 0) {$q_{i}$};
  \node [state, label=below:{$\union{\parallel_{i + 1}}{\roset{\tuple{u,v}}}$}] (qi1) at (6, 0) {$q_{i + 1}$};
  \node (qd2) at (8.25, 0) {$\ldots$};
  \node [state, label=below:{$\parallel_{m}$}] (qm) at (9.75, 0) {$q_{m}$};

  \path [->] (qi) edge node [above] {$\sigma_{i}$} (qi1);
  \path [->] (qi1) edge [loop above] node {$\tau_{i + 1}$} (qi1);

  \node (qphi) at (0, -3) {$\tuple{q, \emptyset}$};
  \node (qd3) at (1.5, -3) {$\ldots$};
  \node (qipi) at (3, -3) {$\tuple{q_{i}, \parallel_{i}}$};
  \node (qi1pi1u) at (6, -3) {$\tuple{q_{i + 1}, \union{\parallel_{i + 1}}{\roset{u,v}}}$};
  \node (qi1pi1) at (6, -5) {$\tuple{q_{i + 1}, \parallel_{i + 1}}$};
  \node (qd4) at (8.25, -3) {$\ldots$};
  \node (qmpm) at (9.75, -3) {$\tuple{q_{m}, \parallel_{m}}$};

  \path [->] (qipi) edge (qi1pi1);
  \path [->>] (qi1pi1) edge (qi1pi1u);

  \begin{pgfonlayer}{background}
    \node [fill=black!30, fit=(qphi) (qd3) (qipi) (qi1pi1u) (qi1pi1) (qd4) (qmpm)] (RegSep) {};
  \end{pgfonlayer}

\end{tikzpicture}
\par\end{centering}

\caption{\label{fig:RegMin:RegSepGraph} The register separation graph. String
$\sigma_{i}$ ``renames'' the separation $\parallel_{i}$ into $\parallel_{i+1}$,
and cycle $\tau_{i+1}$ creates a separation between $u$ and $v$,
while preserving all previously created separations. The goal is to
reach a separation $\parallel_{m}$ which has a $k$-clique of live
registers.}
\end{figure}

\begin{definition}
\label{defn:RegMin:Pspace:SepRel}  Consider some $\acra$ $M$,
and let $\parallel\in\powerset{\cart VV}$ be a relation over $V$.
We say that $\parallel$ separates $q$ if for every constant $c\in\N$,
there exists a string $\sigma$ so $\funcapptrad{\delta}{q_{0},\sigma}=q$
and for all $u,v\in V$, if $u\parallel v$, then $\left|\funcapptrad{\autobox{val}}{u,\sigma}-\funcapptrad{\autobox{val}}{v,\sigma}\right|\geq c$.
Also, we say that a string $\sigma$ $c$-separates $\tuple{q,\parallel}$,
if $\funcapptrad{\delta}{q_{0},\sigma}=q$, and for every $\tuple{u,v}\in\parallel$,
$\left|\funcapptrad{\autobox{val}}{u,\sigma}-\funcapptrad{\autobox{val}}{v,\sigma}\right|\geq c$.
\end{definition}

\begin{definition}
\label{defn:RegMin:Pspace:RegSepGraph}  Consider the set $\powerset{\cart VV}$
of relations over $V$. The \emph{register separation graph} has nodes
$\union{\cart Q{\powerset{\cart VV}}}{\roset t}$, and the following
edges (figure \ref{fig:RegMin:RegSepGraph}):
\begin{enumerate}
\item (Separation edges). From $\tuple{q,\parallel}$ to $\tuple{q,\union{\parallel}{\roset{\tuple{u,v}}}}$
if there is a cycle $\sigma$ at $q$ so that $\funcapptrad{\mu}{q,\sigma,u}=\tuple{u,c}$,
$\funcapptrad{\mu}{q,\sigma,v}=\tuple{v,c^{\prime}}$, $c\neq c^{\prime}$,
and for each $w$ in the domain or range of $\parallel$, $\funcapptrad{\mu}{q,\sigma,w}=\tuple{w,c_{w}}$,
for appropriate $c_{w}\in\Z$.
\item (Renaming edges). From $\tuple{q,\parallel}$ to $\tuple{q^{\prime},\parallel^{\prime}}$
if for some string $\sigma$ that leads $q$ to $q^{\prime}$, whenever
$\tuple{u,v}\in\parallel^{\prime}$, $\funcapptrad{\mu}{q,\sigma,u}=\tuple{u^{\prime},c}$
and $\funcapptrad{\mu}{q,\sigma,v}=\tuple{v^{\prime},c^{\prime}}$,
and $u^{\prime}\parallel v^{\prime}$.
\item (Final edges). From $\tuple{q,\parallel}$ to $t$, if there is a
collection $U\subseteq V$ of $k$ registers, $\left|U\right|=k$,
so for each distinct pair $u,v\in U$, $u\parallel v$.
\end{enumerate}
\end{definition}
Informally, a separation edge identifies a cycle $\tau$ which increments
a pair of registers $u$, $v$ differently, while all other relevant
registers flow into themselves. Renaming edges effect a ``renaming''
of the separation $\parallel$ at $q$ into a separation $\parallel^{\prime}$
at $q^{\prime}$. Final edges to the sink node $t$ exist simply to
identify a uniform target vertex. They are triggered only from vertices
where $k$-separation has already been achieved.

The algorithm is to find a path through the register separation graph
from $\tuple{q_{0},\emptyset}$ to $t$. We first show that a path
exists in the register separation graph from $\tuple{q_{0},\emptyset}$
to $\tuple{q,\parallel}$ iff $\parallel$ separates $q$. But since
the presence of a single edge in this graph can be determined in polynomial
space, and the ``current node'' can be stored in $\funcapptrad O{\left|V\right|^{2}\log\left|Q\right|}$
space, the presence of such a path can also be determined in polynomial
space. Lemmas \ref{lem:RegMin:Pspace:SeparationPossible}, \ref{lem:RegMin:Pspace:PathExists},
and \ref{lem:RegMin:Pspace:Edge} are the three steps to show the
correctness of this approach.

\subsubsection{Connecting $k$-separability to register separation graphs \label{ssub:RegMin:Pspace:Connection}}
\begin{lemma}
\label{lem:RegMin:Pspace:SeparationPossible} If there is a path $\pi$
from $\tuple{q_{0},\emptyset}$ to $\tuple{q,\parallel}$ in the register
separation graph, then $\parallel$ separates $q$.\end{lemma}
\begin{proof}
Informally, since every register pair $\tuple{u,v}\in\opname{\parallel}$
are separated by some separation edge in $\pi$, and no subsequent
edge results in the resetting of this difference (though they might
increase or decrease the difference), the cycle can be passed enough
times to create a sufficiently large separation.

Say there are $m$ separation edges in $\pi$. Then by definition,
for every vector $\mathbf{x}\in\N^{m}$, there is a string $\sigma$
to $q$ so that for all $u\parallel v$, 
\begin{align*}
\funcapptrad{\autobox{val}}{u,\sigma}-\funcapptrad{\autobox{val}}{v,\sigma} & =c_{uv}+\sum_{i}d_{i}^{uv}x_{i},
\end{align*}
where $d_{i}^{uv}$ is the difference created between $u$ and $v$
by the $i^{\mbox{th}}$ separation edge in $\pi$. Also, by construction,
for each $u\parallel v$, there is an $i$ so that $d_{i}^{uv}\neq0$.

If we construct a vector $\mathbf{x}$ so that $\sum_{i}d_{i}^{uv}x_{i}$
are simultaneously non-zero for all $u$, $v$, we are done, for then
by appropriately scaling $\mathbf{x}$, $\funcapptrad{\autobox{val}}{u,\sigma}-\funcapptrad{\autobox{val}}{v,\sigma}$
can be made arbitrarily large in magnitude. Choose $x_{1}=1$, and
once $x_{1}$, \ldots{}, $x_{i}$ are defined, define 
\begin{align}
x_{i+1} & =1+\max_{u\parallel v,d_{i+1}^{uv}\neq0}\left\lceil \frac{\sum_{j\leq i}\left|d_{j}^{uv}\right|x_{j}}{\left|d_{i+1}^{uv}\right|}\right\rceil .
\end{align}
(In the degenerate case when $d_{i+1}^{uv}=0$ for all $u$, $v$,
choose an arbitrary value for $x_{i+1}$) This has the property that
$\left|d_{i+1}^{uv}x_{i+1}\right|>\sum_{j\leq i}d_{j}^{uv}x_{j}$
(if $d_{i+1}^{uv}$ is non-zero), and so $\sum_{i}d_{i}^{uv}x_{i}$
is non-zero for all $u\parallel v$. This completes the proof.\end{proof}
\begin{lemma}
\label{lem:RegMin:Pspace:PathExists} If $\parallel$ separates $q$,
then there is a path through the register separation graph from $\tuple{q_{0},\emptyset}$
to $\tuple{q,\parallel}$.\end{lemma}
\begin{proof}
Consider some pair $\tuple{u,v}\in\opname{\parallel}$ -- since $u$
and $v$ are separable at $q$, intuitively it has to be the case
that there is a cycle $\tau$ resulting in different increments to
$u$ and $v$ (or a path from some other state $q^{\prime}$ where
$u^{\prime}$ and $v^{\prime}$ were differently incremented on $\tau^{\prime}$,
and then the values of these registers flowed into $u$ and $v$ respectively).
We now formalize this intuition

By induction on the number of elements in $\parallel$. There is a
path from $\tuple{q_{0},\emptyset}$ to $\tuple{q,\emptyset}$, for
every reachable state $q$. Say $\parallel$ has $m+1$ elements,
and the proposition holds at every $q$ for every $\parallel^{\prime}$
with at most $m$ elements each. We now show the existence of a reachable
vertex $\tuple{q_{l},\parallel_{l}}$, where $\parallel_{l}$ has
$m$ elements, and there is a path from $\tuple{q_{l},\parallel_{l}}$
to $\tuple{q,\parallel}$.

Consider some state $q^{\prime}$, which on reading symbol $a$ transitions
to $q$. We define the weakest precondition of $\parallel$ with respect
to this transition as the smallest relation $\opname{\parallel^{\prime}}\subseteq\cart VV$
so that whenever $u\parallel v$ then $u^{\prime}\parallel^{\prime}v^{\prime}$,
where $\funcapptrad{\mu}{q^{\prime},a,u}=\tuple{u^{\prime},c_{u}}$
and $\funcapptrad{\mu}{q^{\prime},a,v}=\tuple{v^{\prime},c_{v}}$.
Observe that whenever $\parallel$ separates $q$, there must be a
predecessor state $q^{\prime}$ transitioning to $q$ on some symbol
$a$ so that the weakest precondition of $\parallel$ with respect
to this transition, $\parallel^{\prime}$ separates $q^{\prime}$
(for otherwise, along every path to $q$, because of the unreachability
of the predecessor separation, some registers $u\parallel v$ have
to be close).

Specifically, let $N\subseteq\cart Q{\powerset{\cart VV}}$ be a set
of vertices in the register separation graph. Then, for sufficiently
large $c$, there is a constant $c^{\prime}$ and an $N^{\prime}\subseteq\cart Q{\powerset{\cart VV}}$
of weakest precondition separations so that all strings $\sigma$
that $c$-separate some element of $N$ must be at least one symbol
long, and $\strcat[\ldots]{\sigma_{1}}{\sigma_{\left|\sigma\right|-1}}$
must $c^{\prime}$-separate some element of $N^{\prime}$. If we start
with $N=\roset{\tuple{q,\parallel}}$, and repeat this $n=\left(p+1\right)2^{p}$
times (where $p$ is the number of vertices in the register separation
graph), then some subset $N^{\prime}$ must be repeated at least $p+1$
times, let these positions be $i_{1}$, \ldots{}, $i_{p+1}$, indexed
from the end. Let $c_{n}$ be the separation at $N=\roset{\tuple{q,\parallel}}$
so this process can be repeated $n$ times. Choose the shortest string
$\sigma$ that $c_{n}$-separates $\tuple{q,\parallel}$. (Indexing
$\sigma$ from the end) At least two of $\sigma_{i_{1}}$, \ldots{},
$\sigma_{i_{p+1}}$ must pass through the same state $q_{l}$, and
separate the same subset of registers $\parallel_{l}^{\prime}$. Let
the cycle between these occurrences be $\tau$, so $\sigma=\strcat{\sigma^{\prime}}{\strcat{\tau}{\sigma^{\prime\prime}}}$,
and $\tau\neq\strempty$. For each pair $\tuple{u,v}\in\parallel_{l}^{\prime}$,
consider the register separations after processing $\sigma^{\prime}$
and $\strcat{\sigma^{\prime}}{\tau}$. If no difference changes, then
$\strcat{\sigma^{\prime}}{\sigma^{\prime\prime}}$ also $c_{n}$-separates
$\tuple{q,\parallel}$, contradicting the assumption that $\sigma$
was the shortest such string. Thus, some pair of registers $\tuple{u,v}\in\parallel_{l}^{\prime}$,
must have been incremented differently through this cycle. Define
$\opname{\parallel_{l}}=\opname{\parallel_{l}^{\prime}}\setminus\roset{\tuple{u,v}}$,
so that both edges $\tuple{q_{l},\parallel_{l}}\to\tuple{q_{l},\parallel_{l}^{\prime}}\to\tuple{q,\parallel}$
are present in the register separation graph. $\parallel_{l}$ separates
$q_{l}$, and possesses only $m$ elements. Hence the proof.
\end{proof}

\subsubsection{Putting it all together}
\begin{lemma}
\label{lem:RegMin:Pspace:Edge} Let $\tuple{q,\parallel}$ and $\tuple{q^{\prime},\parallel^{\prime}}$
be nodes in the register separation graph. The problem of determining
whether an edge exists between $\tuple{q,\parallel}$ and $\tuple{q^{\prime},\parallel^{\prime}}$
can be answered in polynomial space.\end{lemma}
\begin{proof}
An edge between two nodes in the register separation graph is either
a cycle edge or a renaming edge. We treat the three cases separately:
\begin{enumerate}
\item Whether a renaming edge exists between $\tuple{q,\parallel}$ and
$\tuple{q^{\prime},\parallel^{\prime}}$ can be done in non-deterministic
polynomial space. We simply guess the witness string $\sigma\in\kstar{\Sigma}$
from $q$ to $q^{\prime}$, one symbol at a time, and update the current
register $q_{t}$ and separation $\parallel_{t}$. We accept if $q_{t}=q^{\prime}$
and $\opname{\parallel^{\prime}}\subseteq\opname{\parallel_{t}}$.
This is essentially a graph-reachability query which is solvable in
$\bigoh{\log\left|Q\right|2^{\left|V\right|^{2}}}$ non-deterministic
space.
\item To determine the presence of a cycle edge, we first observe that it
is an instance of a $1$-counter non-zero reachability problem. A
$1$-counter machine is a tuple $A=\tuple{Q_{A},\delta,q_{0}}$, where
$\delta\subseteq\cart{Q_{A}}{\cart{Q_{A}}{\Z}}$, and $q_{0}\in Q_{A}$.
The semantics are non-deterministic: we start in state $q_{0}$, with
the counter initialized to $0$. If we are currently in a state $q\in Q_{A}$,
then we can transition to any state $q^{\prime}$ so that $\tuple{q,q^{\prime},c}\in\delta$.
During this transition, the counter is incremented by $c$. Given
a final state $q\in Q_{A}$, the non-zero reachability problem asks:
is there a path from $q_{0}$ to $q$ so that the counter value is
non-zero? In our case, the counter encodes the difference between
two registers $u^{\prime}$ and $v^{\prime}$, whose values have been
influenced by the initial values of $u$ and $v$ respectively. The
states $\tuple{q,f}\in Q_{A}$ encode the current state $q\in Q$,
and the current register renaming $\func fVV$, i.e. for each register
$v$, $\funcapptrad fv$ tells us the name of the initial register
whose value has flowed into $v$. Observe that $Q_{A}$ is large:
it has $\funcapptrad O{\left|Q\right|\left|V\right|^{\left|V\right|}}$
states, and thus we never explicitly construct $A$. We recall from
\cite{SST-POPL} that the $1$-counter $0$-reachability problem is
in $\nlogspace$, and can be answered in $\funcapptrad O{\log c\left|Q_{A}\right|}$
non-deterministic space, where $c$ is the largest constant appearing
in the definition of $A$. From this, it follows that the non-zero
reachability problem can also be solved in $\funcapptrad O{\log c\left|Q_{A}\right|}$
non-deterministic space. Thus, the presence of a cycle edge can be
determined in $\funcapptrad O{\log c\left|Q\right|\left|V\right|^{\left|V\right|}}=\funcapptrad O{\log c\left|Q\right|+\left|V\right|\log\left|V\right|}$
non-deterministic space.
\item To determine the presence of a final edge from $\tuple{q,\parallel}$
to $t$, we simply guess the $k$-clique $U$ of separated registers.
This can be done in $\bigoh{\left|V\right|}$ non-deterministic space. 
\end{enumerate}
\end{proof}
We now have the main result of this section:
\begin{theorem}
\label{thm:RegMin:Pspace} Given an $\acra$ $M$ and a number $k$,
there is a $\pspace$ procedure to determine whether its register
complexity is at least $k$.
\end{theorem}
\begin{proof}
We know that the registers of $M$ are $k$-separable iff there is
a path through the register separation graph from $\tuple{q_{0},\emptyset}$
to $t$.

Observe that the register separation graph has $\bigoh{\left|Q\right|2^{\left|V\right|^{2}}}$
nodes. Since graph reachability can be solved in $\nlogspace$, this
problem can be solved in $\bigoh{\log\left|Q\right|+\left|V\right|^{2}}$
non-deterministic space. Putting the procedures together -- separating
loop detection requires $\bigoh{\log c\left|Q\right|+\left|V\right|\log\left|V\right|}$,
renaming edge detection needs $\bigoh{\log\left|Q\right|+\left|V\right|^{2}}$,
and final edge detection needs $\bigoh{\left|V\right|}$ non-deterministic
space. It follows that the register complexity can be determined using
$\bigoh{\log c\left|Q\right|+\left|V\right|^{2}}$ non-deterministic
space.
\end{proof}
An alternative in the above procedure is to use fast polynomial time
algorithms as subroutines: Reachability in a graph with $n$ vertices
can be determined in $\bigoh n$ time, and $1$-counter $0$-reachability
of an $n$ state machine can be decided in $\bigoh{\left(cn\right)^{3}}$
time. With this assumption, the procedure runs in $\bigoh{n\left(n+\left(cn\right)^{3}+2^{\left|V\right|}\left|V\right|^{2}\right)}$
time with $n=\left|Q\right|2^{\left|V\right|^{2}}$, and $c$ is the
largest constant in $M$, giving the final time complexity of the
algorithm as $\bigoh{c^{3}\left|Q\right|^{4}2^{4\left|V\right|^{2}}}$.

\subsection{\label{sub:RegMin:Loop} Pumping lemma for $\acra$s}

The following theorem is the interpretation of a path through the
register separation graph. Given a regular function $f$ of register
complexity at least $k$, it guarantees the existence of $m$ cycles
$\tau_{1}$, \ldots{}, $\tau_{m}$, serially connected by strings
$\sigma_{0}$, \ldots{}, $\sigma_{m}$, so that based on one of $k$
suffixes $w_{1}$, \ldots{}, $w_{k}$, the cost paid on one of the
cycles must differ. These cycles are actually the separation cycles
discussed earlier, and intermediate strings $\sigma_{i}$ correspond
to the renaming edges. Consider for example, the function $f_{2}$
from figure \ref{fig:Prelim:CRA:Examples}, and let $\sigma_{0}=\strempty$,
$\tau_{1}=\strcat a{\strcat ab}$, and $\sigma_{1}=\strempty$. We
can increase the difference between the registers $x$ and $y$ to
arbitrary amounts by pumping cycle $\tau_{1}$. Now if the suffixes
are $w_{1}=a$, and $w_{2}=b$, then the choice of suffix determines
the ``cost'' paid on each iteration of the cycle.
\begin{theorem}
\label{thm:RegMin:Loop} A regular function $\func f{\kstar{\Sigma}}{\Z_{\bot}}$
has register complexity at least $k$ iff there exist strings $\sigma_{0}$,
\ldots{}, $\sigma_{m}$, $\tau_{1}$, \ldots{}, $\tau_{m}$, and
suffixes $w_{1}$, \ldots{}, $w_{k}$, and $k$ distinct coefficient
vectors $\mathbf{c}_{1},\ldots,\mathbf{c}_{k}\in\Z^{m}$ so that for
all vectors $\mathbf{x}\in\N^{m}$, 
\begin{align*}
\funcapptrad f{\strcat{\strcat{\sigma_{0}}{\strcat{\tau_{1}^{x_{1}}}{\strcat{\sigma_{1}}{\strcat[\ldots]{\tau_{2}^{x_{2}}}{\sigma_{m}}}}}}{w_{i}}} & =\sum_{j}c_{ij}x_{j}+d_{i}.
\end{align*}

\end{theorem}
\begin{proof}
We deal with the two cases separately:\end{proof}
\begin{enumerate}
\item If $f$ has register complexity at least $k$, then there is a path
$\pi$ through the register separation graph to a vertex $\tuple{q,\parallel}$
with a $k$-clique of live registers in $\parallel$. Every such path
can be collapsed into one where this is exactly one renaming edge
(possibly corresponding to $\strempty$) between any two cycle edges.
Let $\sigma_{i}$ be the $\left(i+1\right)^{\mbox{th}}$ renaming
edge, and let $\tau_{i}$ be the $i^{\mbox{th}}$ cycle edge. Since
$k$ mutually divergent registers are live, for each such register
$v$, there exists a suffix $w_{v}$ to extract its value. By the
definition of the register separation graph, the claim follows.
\item Say there exist strings $\sigma_{0}$, \ldots{}, $\sigma_{m}$, $\tau_{1}$,
\ldots{}, $\tau_{m}$, $w_{1}$, \ldots{}, $w_{k}$ so that this
holds. Since there are only finitely many states in any given machine
$M$ implementing $f$, there must exist $i_{1}$, $j_{1}$ so that
$\funcapptrad{\delta}{q_{0},\strcat{\sigma_{0}}{\tau_{1}^{i_{1}}}}=\funcapptrad{\delta}{q_{0},\strcat{\sigma_{0}}{\strcat{\tau_{1}^{i_{1}}}{\tau_{1}^{j_{1}}}}}=q_{1}$,
for some $q_{1}\in Q$. Similarly, there must be $i_{2}$, $j_{2}$
so that $\funcapptrad{\delta}{q_{1},\strcat{\sigma_{1}}{\tau_{2}^{i_{2}}}}=\funcapptrad{\delta}{q_{1},\strcat{\sigma_{1}}{\strcat{\tau_{2}^{i_{2}}}{\tau_{2}^{j_{2}}}}}=q_{2}$,
for appropriate $q_{2}$. Repeat this process to reach state $q_{m+1}$.
It now follows that there must exist at least $k$ separable registers
in $q_{m+1}$, since a divergent value is extracted by each $w_{i}$.
Thus, the register complexity of $f$ is at least $k$.\end{enumerate}

\subsection{\label{sub:RegMin:PspaceHard} Computing the register complexity
is $\pspaceh$}

We reduce the DFA intersection non-emptiness checking problem  to
the problem of computing the register complexity. Let $A=\tuple{Q,\Sigma,\delta,q_{0},\roset{q_{f}}}$
be a DFA. Consider a single-state $\acra$ $M$ with input alphabet
$\Sigma$. For each state $q\in Q$, $M$ maintains a register $v_{q}$.
On reading a symbol $a\in\Sigma$, $M$ updates $v_{q}:=v_{\funcapptrad{\delta}{q,a}}$,
for each $q$. Observe that this is simulating the DFA in reverse:
if we start with a special tagged value in $v_{q_{f}}$, then after
processing $\sigma$, that tag is in $v_{q_{0}}$ iff $\strrev{\sigma}$
is accepted by $A$. Also observe that doing this in parallel for
all the DFAs no longer requires an exponential product construction,
but only as many registers as a linear function of the input size.
We use this idea to construct in polynomial time an $\acra$ $M$
whose registers are $\left(k+2\right)$-separable iff there is a string
$\sigma\in\kstar{\Sigma}$ which is simultaneously accepted by all
the DFAs. 

\begin{lemma}
\label{lem:RegMin:PspaceHard:Kozen} The following problem is $\pspacec$
\cite{Kozen1977}: Given a set of DFAs, $\mathcal{A}=\roset{A_{1},\ldots,A_{k}}$
over a common input alphabet $\Sigma$, is the intersection of their
languages non-empty?
\end{lemma}
In particular, the problem remains hard if we restrict the DFAs to
have a single accepting state each, for a DFA over any alphabet could
be extended with a new end-of-string symbol, and made to possess a
single accepting state (incurring only a constant size increase).
\begin{claim}
\label{clm:RegMin:PspaceHard:Kozen2} The following problem is $\pspacec$:
Given a set of DFAs, $\mathcal{A}=\roset{A_{1},\ldots,A_{k}}$ over
a common input alphabet $\Sigma$, and each with a single accepting
state, is the intersection of their languages non-empty?
\end{claim}
In figure \ref{fig:RegMin:PspaceHard}, we describe the reduction
informally. Unlabelled transitions are triggered by special control
symbols not in $\Sigma$. For each state $q$ of each DFA $A_{i}$,
the $\acra$ maintains a register $v_{q}$. Consider the self-loop
in state $q_{1}$ of the separation gadget: on reading symbol $a\in\Sigma$,
each register $v_{q}$ is assigned the value of $v_{\funcapptrad{\delta}{q,a}}$.
Thus, after reading a string $\sigma\in\kstar{\Sigma}$, $v_{q}$
contains the value initially in $v_{\funcapptrad{\delta}{q,\strrev{\sigma}}}$,
where $\strrev{\sigma}$ is the reverse string of $\sigma$. The initial
loop at $q_{0}$ sets up large distinct values in all the final states.
Thus, any string $\sigma$ that is simultaneously accepted by all
DFAs corresponds to a way of reaching $q_{f}$ with large values in
$v_{q_{0i}}$, the registers corresponding to the initial states.
The self-loop at $q_{f}$ sets up a large value in a special register
$u$. Therefore, if the DFAs accept a common string, then $q_{f}$
is $\left(k+2\right)$-separable. If no string is accepted by all
DFAs, then on each path to $q_{f}$, $v_{q_{0i}}=0$, for some $i$,
and hence $q_{f}$ is not $\left(k+2\right)$-separable. Furthermore,
along each path to $q_{0}$ or $q_{1}$, all registers contain one
of at most $k+1$ distinct values, and there is exactly one live register
in each $q_{\autobox{out}i}$. Therefore no state other than $q_{f}$
is $\left(k+2\right)$-separable. Thus, the registers of the separation
gadget are $\left(k+1\right)$-separable iff all the DFAs simultaneously
accept some string.

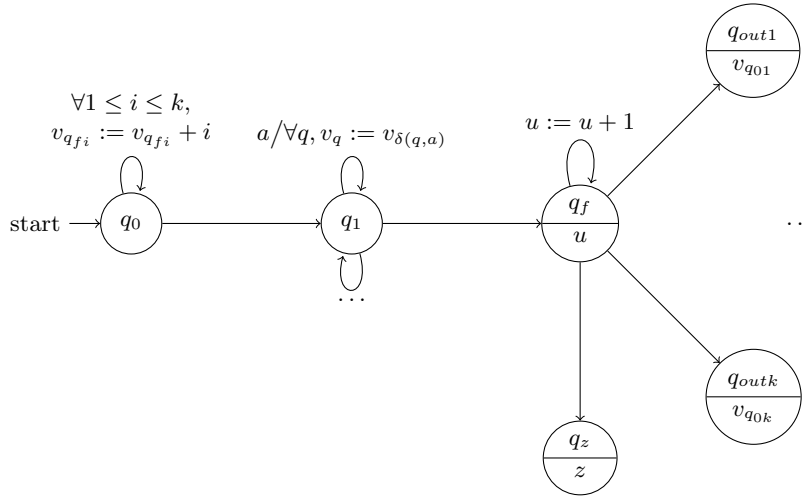
\begin{figure}
\begin{centering}
\pgfmathwidth{"$\quoset{A_{i}}{v_{q_{fi}} := v_{q_{fi}} + 1}$"}
\def\stringswidth{\pgfmathresult}
\pgfmathparse{0.75 * \stringswidth}
\global\edef\stringswidth{\pgfmathresult}

\begin{tikzpicture}[sloped, node distance=\stringswidth pt]

  \node [state, initial] (q0) {$q_{0}$};
  \node [state, right=of q0] (q1) {$q_{1}$};
  \node [state with output, right=of q1] (qf) {$q_{f}$ \nodepart{lower} $u$};
  \node [state with output, below=of qf] (qz) {$q_{z}$ \nodepart{lower} $z$};

  \node [state with output, above right=of qf] (qo1) {$q_{\autobox{out}1}$ \nodepart{lower} $v_{q_{01}}$};
  \node [right=of qf] (qd) {$\ldots$};
  \node [state with output, below right=of qf] (qok) {$q_{\autobox{out}k}$ \nodepart{lower} $v_{q_{0k}}$};

  \path [->] (q0) edge [loop above] node {$\begin{array}{c} \forall 1 \leq i \leq k,\\ v_{q_{fi}} := v_{q_{fi}} + i \end{array}$} (q0);
  \path [->] (q0) edge node {} (q1);
  \path [->] (q1) edge [loop above] node {$\quoset a {\forall q, v_{q} := v_{\funcapptrad{\delta}{q, a}}}$} (q1);
  \path [->] (q1) edge [loop below] node {$\ldots$} (q1);
  \path [->] (q1) edge node {} (qf);
  \path [->] (qf) edge [loop above] node {$u := u + 1$} (qf);
  \path [->] (qf) edge node {} (qz);
  \path [->] (qf) edge node {} (qo1);
  \path [->] (qf) edge node {} (qok);

\end{tikzpicture}
\par\end{centering}

\caption{\label{fig:RegMin:PspaceHard} The separation gadget. In the self-loop
at $q_{1}$, $\delta$ refers to the transition function of the appropriate
DFA.}
\end{figure}

\begin{definition}
\label{defn:RegMin:PspaceHard:SeparationGadget}  Let $\mathcal{A}=\roset{A_{1},\ldots,A_{k}}$
be a set of $k$ DFAs, each with a single accepting state. The \emph{separation
gadget of $\mathcal{A}$} is the following $\acra$ $M=\tuple{Q,\Sigma^{\prime},V,\delta,\mu,q_{0},\nu}$:\end{definition}
\begin{enumerate}
\item $Q=\union{\roset{q_{0},q_{1},q_{f},q_{z}}}{\ruset{q_{\autobox{out}i}}{1\leq i\leq k}}$,
\item $\Sigma^{\prime}=\union{\Sigma}{\union{\roset{\#}}{\ruset{a_{i}}{1\leq i\leq k}}}$,
and
\item $V=\union{\roset{u,z}}{\ruset{v_{q}}{q\in Q_{i},1\leq i\leq k}}$.
\item $\delta$ is defined by the following rules:

\begin{enumerate}
\item $\funcapptrad{\delta}{q_{0},\#}=q_{1}$. For all other $a\in\Sigma^{\prime}$,
$\funcapptrad{\delta}{q_{0},a}=q_{0}$.
\item For each $a\in\Sigma$, $\funcapptrad{\delta}{q_{1},a}=q_{1}$. For
all other $a\in\Sigma^{\prime}$, $\funcapptrad{\delta}{q_{1},a}=q_{f}$.
\item For all $a\in\Sigma$, $\funcapptrad{\delta}{q_{f},a}=q_{f}$. $\funcapptrad{\delta}{q_{f},\#}=q_{z}$.
For each $a_{i}$, $1\leq i\leq k$, $\funcapptrad{\delta}{q_{f},a_{i}}=q_{\autobox{out}i}$.
\item For each $i$, $1\leq i\leq k$, and $a\in\Sigma^{\prime}$, $\funcapptrad{\delta}{q_{\autobox{out}i},a}=q_{\autobox{out}i}$.
\end{enumerate}
\item $\mu$ is defined by the following rules:

\begin{enumerate}
\item For all $a\in\Sigma^{\prime}$ so $\funcapptrad{\delta}{q_{0},a}=q_{0}$,
and $1\leq i\leq k$, $\funcapptrad{\mu}{q_{0},a,v_{q_{fi}}}=\tuple{v_{q_{fi}},i}$.
\item For all $q$, $a\in\Sigma$, $\funcapptrad{\mu}{q_{1},a,v_{q}}=\tuple{v_{\funcapptrad{\delta^{\prime}}{q,a}},0}$.
Here $\delta^{\prime}$ is the transition function of the DFA containing
$q^{\prime}$.
\item For all $a\in\Sigma$, $\funcapptrad{\mu}{q_{f},a,u}=\tuple{u,1}$.
\item For all other $q$, $a$, $v$, $\funcapptrad{\mu}{q,a,v}=\tuple{v,0}$.
\end{enumerate}
\item $\funcapptrad{\nu}{q_{f}}=\tuple{u,0}$, $\funcapptrad{\nu}{q_{z}}=\tuple{z,0}$,
and $\funcapptrad{\nu}{q_{outi}}=\tuple{v_{q_{0i}},0}$, for all $i$.
In all other states, $\funcapptrad{\nu}q=\bot$.\end{enumerate}
\begin{proposition}
\label{clm:RegMin:PspaceHard:Prop} Let $\mathcal{A}$ be a set of
$k$ DFAs, and $M$ be the separation gadget of $\mathcal{A}$.
\begin{enumerate}
\item Let $\sigma\in\kstar{\left(\Sigma^{\prime}\right)}$ so $\funcapptrad{\delta}{q_{0},\sigma}\notin\roset{q_{f},q_{z},q_{\autobox{out}i}}$.
Then there is a collection $P\subseteq\Z$ with $\left|P\right|\leq k+1$,
so for each register $v\in V$, $\funcapptrad{\autobox{val}}{v,\sigma}\in P$.
\item If the intersection language of the DFAs is empty, then for each $\sigma\in\kstar{\left(\Sigma^{\prime}\right)}$,
if $\funcapptrad{\delta}{q_{0},\sigma}=q_{f}$, there is some $i$
so that $\funcapptrad{\autobox{val}}{q_{0i},\sigma}=0=\funcapptrad{\autobox{val}}{z,\sigma}$.
\item If the intersection language of the DFAs is non-empty, then for each
$c\in\Z$, there is a $\sigma\in\kstar{\left(\Sigma^{\prime}\right)}$
so that $\funcapptrad{\delta}{q_{0},\sigma}=q_{f}$, and for each
$v,v^{\prime}\in\union{\roset{u,z}}{\ruset{q_{0i}}{1\leq i\leq k}}$,
$\left|\funcapptrad{\autobox{val}}{v,\sigma}-\funcapptrad{\autobox{val}}{v^{\prime},\sigma}\right|\geq c$.
\end{enumerate}
\end{proposition}
We now conclude the hardness argument:
\begin{theorem}
\label{thm:RegMin:PspaceHard} Given an $\acra$ $M$ and a number
$k$, deciding whether the register complexity of $\bbracket M$ is
at least $k$ is $\pspaceh$.
\end{theorem}
\begin{proof}
Given a set of $k$ DFAs $\mathcal{A}$, the separation gadget $M$
of $\mathcal{A}$ can be constructed in polynomial time ($M$ has
$k+3$ states, $2+\sum_{i}\left|Q_{i}\right|$ registers, and operates
over an alphabet of $k+\left|\Sigma\right|+1$ symbols). From proposition
\ref{clm:RegMin:PspaceHard:Prop}, it follows that an equivalent $\acra$
with $k+1$ registers exists iff the intersection language is empty.
Thus, the problem is $\pspaceh$.\end{proof}

\section{\label{sec:Games} Games over $\acra$s}

We now study games played over $\acra$s. We extend the model of $\acra$s
to allow alternation -- in each state, a particular input symbol may
be associated with multiple transitions. The system picks the input
symbol to process, while the environment picks the specific transition
associated with this input symbol. Accepting states are associated
with output functions, and the system may choose to end the game in
any accepting state. Given a budget $k$, we wish to decide whether
the system has a winning strategy with worst-case cost no more than
$k$.  We show that $\acra$ games are undecidable when the registers
are integer-valued, and $\exptimec$ when the domain is $\D=\N$.
\begin{definition}
\label{defn:Games}  An $\dacra{\D}$ reachability game is played
over a structure $G=\tuple{Q,\Sigma,V,\delta,\mu,q_{0},F,\nu}$, where
$Q$, $\Sigma$, and $V$ are finite non-empty sets of states, input
symbols and registers respectively, $\delta\subseteq\cart Q{\cart{\Sigma}Q}$
is the transition relation, $\func{\mu}{\cart{\delta}V}{\cart V{\D}}$
is the register update function, $q_{0}\in Q$ is the start state,
$F\subseteq Q$ is the set of accepting states, and $\func{\nu}F{\cart V{\D}}$
is the output function.

The game configuration is a tuple $\gamma=\tuple{q,\autobox{val}}$,
where $q\in Q$ is the current state, and $\func{\autobox{val}}V{\D}$
is the current register valuation. A run $\pi$ is a (possibly infinite)
sequence of game configurations $\tuple{q_{1},\autobox{val}_{1}}\to^{a_{1}}\tuple{q_{2},\autobox{val}_{2}}\to^{a_{2}}\cdots$
with the property that
\begin{enumerate}
\item the transition $q_{i}\to^{a_{i}}q_{i+1}\in\delta$ for each $i$,
and
\item $\funcapptrad{\autobox{val}_{i+1}}u=\funcapptrad{\autobox{val}_{i}}v+c$,
where $\funcapptrad{\mu}{q_{i}\to^{a_{i}}q_{i+1},u}=\tuple{v,c}$,
for each register $u$ and transition $i$.
\end{enumerate}
A strategy is a function $\func{\strat}{\cart{\kstar Q}Q}{\Sigma}$
that maps a finite history $\strcat{q_{1}}{\strcat[\ldots]{q_{2}}{q_{n}}}$
to the next symbol $\funcapptrad{\strat}{\strcat{q_{1}}{\strcat[\ldots]{q_{2}}{q_{n}}}}$.
A run $\pi$ is consistent with $\strat$ if for each $i$, $\funcapptrad{\strat}{\strcat{q_{1}}{\strcat[\ldots]{q_{2}}{q_{i}}}}=a_{i}$.
$\strat$ is winning starting from a state $q$ if for every run $\pi$
consistent with $\strat$ and starting from $q_{1}=q$, there is some
$i$ so that $q_{i}\in F$. It is winning from a configuration $\tuple{q,\autobox{val}}$
with a budget of $\bud\in\D$ if for every consistent run $\pi$ starting
from $\tuple{q_{1},\autobox{val}_{1}}=\tuple{q,\autobox{val}}$, for
some $i$, $q_{i}\in F$ and $\funcapptrad{\nu}{q_{i},\autobox{val}_{i}}\leq\bud$.
\end{definition}
For greater readability, we write tuples $\tuple{q,a,q^{\prime}}\in\delta$
as $q\to^{a}q^{\prime}$. If $q\in F$, and $\autobox{val}$ is a
register valuation, we write $\funcapptrad{\nu}{q,\autobox{val}}$
for the result $\funcapptrad{\autobox{val}}v+c$, where $\funcapptrad{\nu}q=\tuple{v,c}$.
When we omit the starting configuration for winning strategies it
is understood to mean the initial configuration $\tuple{q_{0},\autobox{val}_{0}}$
of the $\acra$.

Consider the natural partial order $\preceq$ over register valuations:
$\autobox{val}\preceq\autobox{val}^{\prime}$ iff for all registers
$v$, $\funcapptrad{\autobox{val}}v\leq\funcapptrad{\autobox{val}^{\prime}}v$.
Then, any winning strategy for large valuations is also a winning
strategy for small valuations:
\begin{claim}
\label{clm:Games:ValuationPoset} For each $q$, $\bud$, $\autobox{val}$,
$\autobox{val}^{\prime}$, if $\autobox{val}\preceq\autobox{val}^{\prime}$,
then every strategy $\strat$ which is $\bud$-winning starting from
$\tuple{q,\autobox{val}^{\prime}}$ is also $\bud$-winning starting
from $\tuple{q,\autobox{val}}$.\end{claim}

\subsection{\label{sub:Games:Exptime-N} $\dacra{\N}$ reachability games can
be solved in $\exptime$}

Consider the simpler class of (unweighted) graph reachability games.
These are played over a structure $G^{f}=\tuple{Q,\Sigma,\delta,q_{0},F}$,
where $Q$ is the finite state space, and $\Sigma$ is the input alphabet.
$\delta\subseteq\cart Q{\cart{\Sigma}Q}$ is the state transition
relation, $q_{0}\in Q$ is the start state, and $F\subseteq Q$ is
the set of accepting states. If the input symbol $a\in\Sigma$ is
played in a state $q$, then the play may adversarially proceed to
any state $q^{\prime}$ so that $\tuple{q,a,q^{\prime}}\in\delta$.
The system can force a win if every run compatible with some strategy
$\func{\strat^{f}}{\cart{\kstar Q}Q}{\Sigma}$ eventually reaches
a state $q_{f}\in F$. Such games can be solved by a recursive back-propagation
algorithm -- corresponding to model checking the formula $\mu X\cdot\left(F\lor\bigvee_{a\in\Sigma}\left[a\right]X\right)$
-- in time $\bigoh{\left|Q\right|\left|\Sigma\right|}$. Observe that
these games obey the ``small strategy'' property: if there is a
winning strategy $\strat$, then there is a winning strategy $\strat_{\autobox{small}}$
which guarantees that no state is visited twice.

From every $\dacra{\N}$ reachability game $G=\tuple{Q,\Sigma,V,\delta,\mu,q_{0},F,\nu}$,
we can project out an unweighted graph reachability game $G^{f}=\tuple{Q,\Sigma,\delta,q_{0},F}$.
Also, $G^{f}$ has a winning strategy iff for some $\bud\in\N$, $G$
has a $k$-winning strategy. Consider the cost of $\strat_{\autobox{small}}$
(computed for $G^{f}$) when used with $G$. Since no run ever visits
the same state twice, $\strat_{\autobox{small}}$ is $c_{0}\left|Q\right|$-winning,
where $c_{0}$ is the largest constant appearing in $G$. We have
thus established an upper-bound on the optimal reachability strategy,
if it exists.

Now assume that we are given an upper-bound $\bud$, and asked to
determine whether a winning strategy $\strat$ exists within this
budget. Because the register increments are non-negative, once a register
$v$ achieves a value larger than $\bud$, it cannot contribute to
the final output, on any suffix $\sigma$ permitted by the winning
strategy. We thus convert $G$ into an unweighted graph reachability
$G_{\bud}^{f}$, where the value of each register is explicitly tracked
in the state, until it is larger than $\bud$. After this, its value
is clamped down to $\bud+$1. 

\begin{definition}
\label{defn:Games:Exptime-N:Gadget} Let $G=\tuple{Q,\Sigma,V,\delta,\mu,q_{0},F,\nu}$
be an $\dacra{\N}$ reachability game. Then, for $\bud\in\N$, define
the corresponding graph reachability game 
\begin{align*}
G_{\bud}^{f} & =\tuple{Q^{\prime}=\cart Q{\left[\bud+1\right]^{\left|V\right|}},\Sigma,\delta^{\prime},\tuple{q_{0},\vector 0},F}
\end{align*}
as follows. Here $\left[\bud+1\right]=\roset{0,1,2,\ldots,\bud+1}$.
Consider some state $\tuple{q,\vector x}\in Q^{\prime}$, and $a\in\Sigma$.
Define $\func{\autobox{val}_{\vector x}}V{\N}$ as $\funcapptrad{\autobox{val}_{\vector x}}u=x_{u}$.
Say also that $\tuple{q,\autobox{val}}\to^{a}\tuple{q^{\prime},\autobox{val}^{\prime}}$,
for some $q^{\prime}$, $\autobox{val}^{\prime}$ is a valid transition
of the game configuration of $G$ on playing symbol $a$. Define $\vector x_{\autobox{val}}^{\prime}$
as $x_{\autobox{val},u}^{\prime}=\funcapptrad{\autobox{val}^{\prime}}u$,
if $\funcapptrad{\autobox{val}^{\prime}}u\leq\bud$. Otherwise $x_{\autobox{val},u}^{\prime}=\bud+1$.
Then $\left(\tuple{q,\vector x}\to^{a}\tuple{q^{\prime},\vector x_{\autobox{val}}^{\prime}}\right)\in\delta^{\prime}$.
Define $\tuple{q,\vector x}\in F\iff\funcapptrad{\nu}{q,\autobox{val}_{\vector x}}\leq\bud$.
\end{definition}
We claim that $G$ has a $\bud$-winning strategy $\strat$ iff the
player can force a win in $G_{\bud}^{f}$. Consider any state $\tuple{q,\vector x}\in Q^{\prime}$
from which the player can force a win. By induction on the assertion
that $\tuple{q,\vector x}$ is winning, we can show there is a $\bud$-winning
strategy from every configuration $\tuple{q,\autobox{val}}$ in $G$
where $\vector x=\vector x_{\autobox{val}}$. Conversely, pick a configuration
$\tuple{q,\autobox{val}}$ of $G$ from which a $\bud$-winning strategy
$\strat$ exists. It follows that $\strat$ is also a winning strategy
in $G_{\bud}^{f}$.

Furthermore, the decision procedure for this problem can be translated
into an optimization procedure: given an upper bound on the budget
$\bud$, determine the smallest $\bud^{\prime}\leq\bud$, if exists,
so that $G$ has a $\bud^{\prime}$-winning strategy. From our discussion
in the main paper, we know that if a winning strategy exists, then
there is a winning strategy $\strat$ with budget at most $c_{0}\left|Q\right|$.
Hence we have:
\begin{theorem}
\label{thm:Games:Exptime-N} The optimal strategy $\strat$ for an
$\dacra{\N}$ reachability game $G$ can be computed in time $\bigoh{\left|Q\right|\left|\Sigma\right|2^{\left|V\right|\log c_{0}\left|Q\right|}}$,
where $c_{0}$ is the largest constant appearing in the description
of $G$.
\end{theorem}
Note that the optimal strategy in $\dacra{\N}$ games need not be
memoryless: we might want to return to a state with a different register
valuation. However, the strategy $\strat$ constructed in the proof
of the above theorem is memoryless given the pair $\tuple{q,\autobox{val}}$
of the current state and register valuation.

\subsection{\label{sub:Games:Hardness} Hardness of solving $\dacra{\D}$ reachability
games}

\subsubsection{\label{ssub:Games:ExptimeHard-N} $\dacra{\N}$ games are $\exptimeh$}

We reduce the halting problem for linearly bounded alternating Turing
machines to the problem of determining a winning strategy in an $\dacra{\N}$
reachability game.
\begin{definition}
\label{defn:Games:ExptimeHard-N:LBAltTM} A linearly bounded alternating
Turing machine is a tuple $M=\tuple{Q=\union{Q_{\lor}}{Q_{\land}},\Gamma,\delta,q_{0},F,n}$.
$Q$ is the state space which is partitioned into ``or''-states
$Q_{\lor}$ and and-states $Q_{\land}$. $\Gamma=\roset{0,1}$ is
a binary tape alphabet, and $\func{\delta}{\cart Q{\cart{\Gamma}{\roset{1,2}}}}{\cart Q{\cart{\Gamma}{\roset{L,R}}}}$
is the transition function. $q_{0}\in Q$ is the initial state, and
$F\subseteq Q$ is the set of accepting states. $n\in\N$ is the length
of the tape, specified in unary.

The configuration is a tuple $\gamma=\tuple{q,\sigma,\autobox{pos}}$,
where $q\in Q$ is the current state, $\sigma\in\Gamma^{n}$ is the
tape string, and $\autobox{pos}\in\roset{1,2,\ldots,n}$ is the position
of the tape head. The initial configuration is $\tuple{q_{0},0^{n},1}$.
In each configuration $\tuple{q,\sigma,\autobox{pos}}$, $\delta$
identifies two successors, corresponding to $\funcapptrad{\delta}{q,\sigma_{\autobox{pos}},1}$
and $\funcapptrad{\delta}{q,\sigma_{\autobox{pos}},2}$ respectively.
Starting from a configuration $\tuple{q,\sigma,\autobox{pos}}$, the
machine $M$ eventually halts if either:
\begin{enumerate}
\item $q\in F$ is an accepting state, or
\item $q\in Q_{\lor}$ is an or-state and at least one of its successor
configurations eventually halts, or
\item $q\in Q_{\land}$ is an and-state and both its successor configurations
eventually halt.
\end{enumerate}
\end{definition}
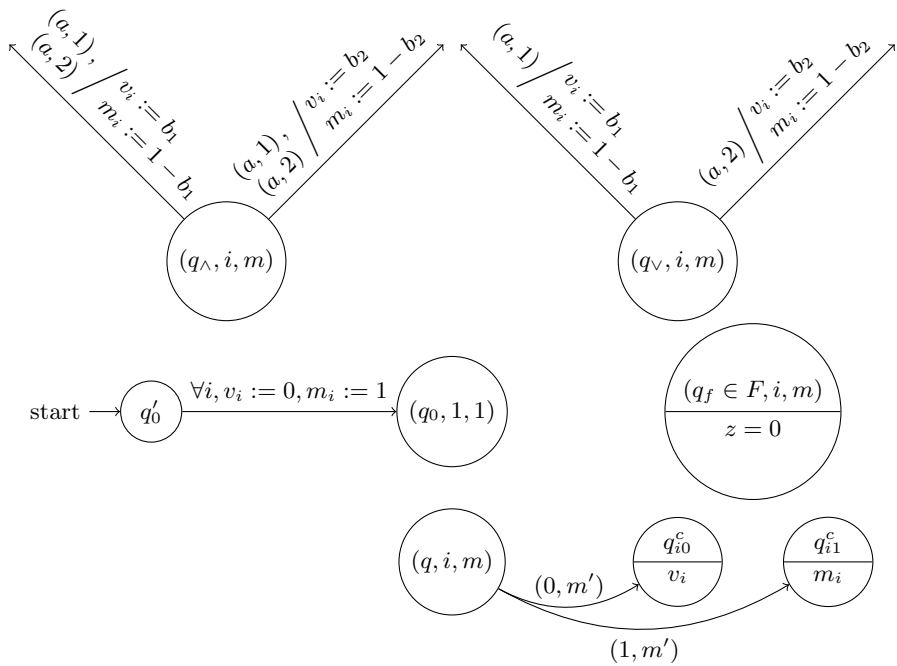
\begin{figure}
\begin{centering}
\begin{tikzpicture}

  \node [state, initial] (q0p) at (0, 0) {$q_{0}^{\prime}$};
  \node [state] (q0_1_1) at (4, 0) {$\tuple{q_{0}, 1, 1}$};
  \node [state with output] (qf_i_m) at (8, 0) {$\tuple{q_f \in F, i, m}$ \nodepart{lower} $z = 0$};
  \node [state with output] (qc_i_0) at (7, -2) {$q^{c}_{i0}$ \nodepart{lower} $v_{i}$};
  \node [state with output] (qc_i_1) at (9, -2) {$q^{c}_{i1}$ \nodepart{lower} $m_{i}$};

  \node [state] (q_i_m) at (4, -2) {$\tuple{q, i, m}$};
  \node [state] (qn_i_m) at (1, 2) {$\tuple{q_{\land}, i, m}$};
  \node [state] (qu_i_m) at (7, 2) {$\tuple{q_{\lor}, i, m}$};

  \node (qn_1) at (-2, 5) {};
  \node (qn_2) at (4, 5) {};
  \node (qu_1) at (4, 5) {};
  \node (qu_2) at (10, 5) {};

  \path [->] (q0p) edge [above] node {$\forall i, v_{i} := 0, m_{i} := 1$} (q0_1_1);
  \path [->] (q_i_m) edge [bend right, above] node {$\tuple{0, m^{\prime}}$} (qc_i_0);
  \path [->] (q_i_m) edge [bend right, below] node {$\tuple{1, m^{\prime}}$} (qc_i_1);

  \path [->] (qn_i_m) edge [above, sloped] node
      {$\quoset {\begin{array}{c}\tuple{a,1},\\\tuple{a,2}\end{array}} {\begin{array}{l}v_{i}:=b_{1}\\m_{i}:=1-b_{1}\end{array}}$}
      (qn_1);
  \path [->] (qn_i_m) edge [above, sloped] node
      {$\quoset {\begin{array}{c}\tuple{a,1},\\\tuple{a,2}\end{array}} {\begin{array}{l}v_{i}:=b_{2}\\m_{i}:=1-b_{2}\end{array}}$}
      (qn_2);
  \path [->] (qu_i_m) edge [above, sloped] node {$\quoset {\tuple{a, 1}} {\begin{array}{l}v_{i}:=b_{1}\\m_{i}:=1-b_{1}\end{array}}$} (qu_1);
  \path [->] (qu_i_m) edge [above, sloped] node {$\quoset {\tuple{a, 2}} {\begin{array}{l}v_{i}:=b_{2}\\m_{i}:=1-b_{2}\end{array}}$} (qu_2);

\end{tikzpicture}
\par\end{centering}

\caption{\label{fig:Games:ExptimeHard-N} Halting gadget $G_{M}$ for a linearly
bounded alternating Turing machine $M$.}
\end{figure}

We construct the gadget shown in figure \ref{fig:Games:ExptimeHard-N}.
There are two types of states: configuration states of the form $\tuple{q,i,m}$
indicating that the TM is in state $q$, the tape head is in position
$i$, and the last choice was move $m\in\roset{1,2}$, and challenge
states of the form $q_{ia}^{c}$ challenging the system to show that
the symbol at position $i$ of the tape is $a$. For each position
$i$ of the tape, we maintain two registers $v_{i}$, $m_{i}$. We
maintain the invariant that $v_{i}=a=1-m_{i}$. Observe that to each
state $\tuple{q,i,m}$ and input symbol $a\in\Gamma$ indicating the
current symbol under the head, there are two successors. If $q=q_{\land}$
is an and-state, then regardless of $m^{\prime}$, on processing $\tuple{a,m^{\prime}}$,
either transition may be taken. If $q=q_{\lor}$ is an or-state, then
on processing $\tuple{a,m^{\prime}}$, we transition to state $\tuple{q^{\prime},j,m^{\prime}}$.
Here $q^{\prime}$, $j$ are respectively the next state and next
tape head position. On each transition, the tape symbol registers
$v_{i}$, $m_{i}$ are appropriately updated. We now formalize:
\begin{definition}
\label{defn:Games:ExptimeHard-N:Gadget} Let $M=\tuple{Q=\union{Q_{\lor}}{Q_{\land}},\Gamma,\delta,q_{0},F,n}$
be a linearly bounded alternating Turing machine. Construct the following
$\dacra{\N}$ reachability game $G_{M}=\tuple{Q^{\prime},\Sigma,V,\delta^{\prime},\mu,q_{0}^{\prime},F^{\prime},\nu}$.
\begin{enumerate}
\item $Q^{\prime}=\union{\roset{q_{0}^{\prime}}}{\union{\left(\cart Q{\cart{\left[n\right]}{\roset{1,2}}}\right)}{\ruset{q_{i,a}^{c}}{\forall i\in\left[n\right],a\in\Gamma}}}$.
\item $\Sigma=\cart{\Gamma}{\roset{1,2}}$.
\item $V=\union{\ruset{v_{i},m_{i}}{\forall i\in\left[n\right]}}{\roset z}$.
\item Define $\delta^{\prime}$ as follows. For all symbols $a\in\Sigma$,
$\left(q_{0}^{\prime}\to^{a}\left(q_{0},1,1\right)\right)\in\delta^{\prime}$.

\begin{enumerate}
\item Let $q\in Q_{\lor}$, $a\in\Gamma$, $m,m^{\prime}\in\roset{1,2}$,
and $i\in\left[n\right]$. Say that $\funcapptrad{\delta}{q,a,m}=\tuple{q^{\prime},b,d}$.
If executing this transition with the tape head at $i$ leads to it
being at position $j$, then $\tuple{q,i,m^{\prime}}\to^{\left(a,m\right)}\tuple{q^{\prime},j,m}\in\delta^{\prime}$
and $\tuple{q,i,m^{\prime}}\to^{\tuple{a,m}}q_{i,a}^{c}\in\delta^{\prime}$.
\item Let $q\in Q_{\land}$, $a\in\Gamma$, $m,m^{\prime},m^{\prime\prime}\in\roset{1,2}$,
and $i\in\left[n\right]$. Say that $\funcapptrad{\delta}{q,a,m}=\tuple{q^{\prime},b,d}$.
Let $j$ be the new head position, then $\tuple{q,i,m^{\prime}}\to^{\left(a,m^{\prime\prime}\right)}\tuple{q^{\prime},j,m}\in\delta^{\prime}$
and $\tuple{q,i,m^{\prime}}\to^{\tuple{a,m^{\prime\prime}}}q_{i,a}^{c}\in\delta^{\prime}$.
\end{enumerate}
\item Define $\mu$ as follows. For all transitions $\tau=\tuple{q_{0}^{\prime}\to^{a}q}\in\delta$,
$\funcapptrad{\mu}{\tau,v_{i}}=\tuple{z,0}$, and $\funcapptrad{\mu}{\tau,m_{i}}=\tuple{z,1}$,
for all $i$. Let $\tau=\tuple{q,i,m}\to^{\tuple{a,m^{\prime}}}\tuple{q^{\prime},j,m^{\prime\prime}}$
be some transition in $\delta^{\prime}$. Let $b$ be the tape symbol
left behind by $\funcapptrad{\delta}{q,a,m^{\prime\prime}}$. Then
define $\funcapptrad{\mu}{\tau,v_{i}}=\tuple{z,b}$ and $\funcapptrad{\mu}{\tau,m_{i}}=\tuple{z,1-b}$.
For all other transitions $\tau\in\delta^{\prime}$ and registers
$v\in V$, define $\funcapptrad{\mu}{\tau,v}=\tuple{v,0}$.
\item Define $F^{\prime}=\union{\ruset{\tuple{q,i,m}}{q\in F}}{\ruset{q_{i,0}^{c},q_{i,1}^{c}}{1\leq i\leq n}}$.
For all $i$, $m$, $\funcapptrad{\nu}{q,i,m}=\tuple{z,0}$. $\funcapptrad{\nu}{q_{i,0}^{c}}=v_{i}$
and $\funcapptrad{\nu}{q_{i,1}^{c}}=m_{i}$, for all $i$.
\end{enumerate}
\end{definition}
Here $\left[n\right]=\roset{1,2,\ldots,n}$, and $Q_{c}=\ruset{q_{i,a}^{c}}{\forall i,a}$
are the challenge states. $z$ is the constant register, always holding
the value $0$. By induction on the assertion that the starting configuration
$\tuple{q,\sigma,i}$ of the TM eventually halts, we have:
\begin{claim}
\label{clm:Games:ExptimeHard-N:Halting--Strategy} Let $\tuple{q,\sigma,i}$
be a configuration starting from which $M$ eventually halts. Then,
for each $m$, there is a $0$-winning strategy $\strat$ in $G_{M}$
starting from $\tuple{\tuple{q,i,m},\autobox{val}}$, where $\autobox{val}$
encodes $\sigma$.
\end{claim}
Let $\strat$ be a $0$-winning strategy in $G_{M}$, and consider
its strategy tree. At some internal node, let it issue input symbol
$\tuple{a_{k},m_{k}}$, and let 
\begin{align*}
\pi & =q_{0}^{\prime}\to^{\tuple{a_{0},m_{0}}}\tuple{q_{1},i_{1},n_{1}}\to^{\tuple{a_{1},m_{1}}}\ldots\to^{\tuple{a_{k-1},m_{k-1}}}\tuple{q_{k},i_{k},m_{k}}
\end{align*}
be the prefix of the run leading up to this node. It follows by induction
on $\pi$ that $a_{k}$ is the current symbol under the tape head
on the appropriate run of $M$ (otherwise the adversary can lead the
player to the challenge state $q_{i_{k},a_{k}}^{c}$, but we assumed
that $\strat$ was a $0$-winning strategy). Since $\strat$ is $0$-winning,
every leaf of its decision tree must point to an accepting state.
Furthermore, any winning strategy has to be associated with a finite
decision tree, it follows that every run of $M$ is accepting. Thus,
\begin{claim}
\label{clm:Games:ExptimeHard-N:Strategy--Halting} If there is a $0$-winning
strategy $\strat$ in $G_{M}$, then $M$ eventually halts.
\end{claim}
Note that $Q^{\prime}$ has $\bigtheta{\left|Q\right|n}$ elements,
$\Sigma$ has $\bigtheta 1$ elements, and $V$ has $\bigtheta n$
registers, where the tape size $n$ was specified in unary. So $G_{M}$
can be constructed in polynomial time given $M$. We thus conclude
our argument:
\begin{theorem}
\label{thm:Games:ExptimeHard-N} Determining whether there is a winning
strategy with budget $\bud$ in an $\dacra{\N}$ reachability game
is $\exptimeh$.

\begin{remark}
Note that $G_{M}$ never really needs to increment any register, since
during all transitions, the values are either maintained unchanged,
or reset from the constant register $z$. This suggests that the hardness
comes from the combinatorial structure of the game rather than the
specific grammar that allows increments.\end{remark}

\end{theorem}

\subsubsection{\label{ssub:Games:Undecidable-Z} Undecidability of $\dacra{\Z}$
reachability games}

\global\long\def\twocminc{\autobox{{\tt inc}}}

\global\long\def\twocmdec{\autobox{{\tt dec}}}

\global\long\def\twocmgoto#1#2#3{{\tt if}\mbox{ }#1\geq0\mbox{ }{\tt goto}\mbox{ }#2\mbox{ }{\tt else}\mbox{ }{\tt goto}\mbox{ }#3}

\global\long\def\twocmhalt{\autobox{{\tt halt}}}

We reduce the halting problem for two-counter machines to the problem
of solving a $\dacra{\Z}$ reachability game. A two-counter machine
$M$ is a sequence of commands $L=\roset{l_{1},l_{2},\ldots,l_{n}}$,
where each command is of the form $\funcapptrad{\twocminc}c$, $\funcapptrad{\twocmdec}c$,
$\twocmgoto c{l_{1}}{l_{2}}$, or $\twocmhalt$, where $c$ refers
to one of the counters $\roset{c_{1},c_{2}}$, and $l_{1},l_{2}\in L$
is the next location. Both counters are integer-valued and initialized
to $0$, and machine execution proceeds sequentially starting from
location $l_{1}$. The semantics of these machines are standard, and
we will not formally define them.

As with our earlier $\exptimeh$ness proof, the gadget $G_{M}$ we
construct has $4$ registers $v_{1}$, $m_{1}$, $v_{2}$, $m_{2}$.
Registers $v_{1}=-m_{1}$ maintain the value of counter $c_{1}$,
while registers $v_{2}=-m_{2}$ maintain the value of counter $c_{2}$.
The challenge states $q_{c<0}$, $q_{c\geq0}$ for $c\in\roset{c_{1},c_{2}}$
force the system to prove the appropriate assertion about the counter
value. The rest of the states are simply the locations $L$ of the
two-counter machine. In location $l\in L$, the system proposes the
input symbol $\tuple{a,b}\in\cart{\roset{c_{1}<0,c_{1}\geq0}}{\roset{c_{2}<0,c_{2}\geq0}}$.
Each component of the tuple is an assertion about the value of the
respective counter. Control proceeds to the next location $l^{\prime}$
depending on the location at $l$ and the input symbol just received.
The counters are incremented / decremented appropriately. We show
that $G_{M}$ has a $0$-winning strategy $\strat$ iff $M$ eventually
halts.
\begin{definition}
\label{defn:Games:Undecidable-Z:Gadget} Let $M$ be a two-counter
machine. Then, the halting gadget $G_{M}=\tuple{Q,\Sigma,V,\delta,\mu,l_{1},\nu}$
is the following $\dacra{\Z}$ reachability game.
\begin{enumerate}
\item $Q=\union L{\roset{q_{c_{1}<0},q_{c_{1}\geq0},q_{c_{2}<0},q_{c_{2}\geq0}}}$.
We refer to the special states $Q_{c}=\roset{q_{c_{1}<0},q_{c_{1}\geq0},q_{c_{2}<0},q_{c_{2}\geq0}}$
as the challenge states.
\item $\Sigma=\cart{\roset{c_{1}<0,c_{1}\geq0}}{\roset{c_{2}<0,c_{2}\geq0}}$.
\item $V=\roset{v_{1},m_{1},v_{2},m_{2},z}$.
\item Define the transition relation $\delta$ as follows. Let $l_{i}$,
$l_{j}$ be arbitrary program locations so $l_{j}$ can follow $l_{i}$
in execution. Then 

\begin{enumerate}
\item Let $l_{i}$ be either an increment or decrement instruction. Then
$\left(l_{i}\to^{a}l_{i+1}\right)\in\delta$, for each $a\in\Sigma$.
\item Let $l_{i}$ be the instruction $\twocmgoto{c_{1}}l{l^{\prime}}$.
Then, the transitions $l_{i}\to^{\tuple{c_{1}\geq0,a}}l$, $l_{i}\to^{\tuple{c_{1}\geq0,a}}q_{c_{1}\geq0}$,
$l_{i}\to^{\tuple{c_{1}<0,a}}l^{\prime}$, $l_{i}\to^{\tuple{c_{1}<0,a}}q_{c_{1}<0}$
occur in $\delta$. And similarly for the conditional jumps on $c_{2}$.
\end{enumerate}
\item Define the register update function $\mu$ as follows. Let $l_{i}$
be instruction $\funcapptrad{\twocminc}{c_{1}}$. Then $\funcapptrad{\mu}{l_{i}\to^{a}l_{i+1},v_{1}}=\tuple{v_{1},1}$,
and $\funcapptrad{\mu}{l_{i}\to^{a}l_{i+1},m_{1}}=\tuple{m_{1},-1}$.
On $\funcapptrad{\twocmdec}{c_{1}}$, $v_{1}$ is decremented and
$m_{1}$ is incremented. Similarly for $\funcapptrad{\twocminc}{c_{2}}$
and $\funcapptrad{\twocmdec}{c_{2}}$. In all other transitions $\tau$,
$\funcapptrad{\mu}{\tau,v}=\tuple{v,0}$, for each register $v$,
i.e. the register is left unchanged.
\item For all halting locations $l=\twocmhalt\in L$, define $\funcapptrad{\nu}l=\tuple{z,0}$.
For all non-halting locations $l\in L$, define $\funcapptrad{\nu}l=\bot$.
For the challenge states, define $\funcapptrad{\nu}{q_{c_{i}<0}}=v_{i}+1$,
$\funcapptrad{\nu}{q_{c_{i}\geq0}}=m_{i}$.\end{enumerate}
\end{definition}

\begin{theorem}
\label{thm:Games:Undecidable-Z} Determining whether there is a winning
strategy with budget $\bud$ in an $\dacra{\Z}$ reachability game
is undecidable.
\end{theorem}
\begin{proof}
We establish the claim that $M$ halts iff $G_{M}$ permits a winning
strategy $\gamma$ with budget $0$. The undecidability of solving
$\dacra{\Z}$ reachability games follows from the undecidability of
the halting problem for two-counter machines \cite{Minsky-2CM}.

First, assume that $M$ halts. We want to construct a winning strategy
$\gamma$ with budget $0$. Consider the finite execution of $M$.
After executing the first $k$ steps, the player issues the symbol
$\tuple{s_{1},s_{2}}$, where $s_{1}$, $s_{2}$ are respectively
the signs of the values in $c_{1}$, $c_{2}$ after the two-counter
machine executes for $k$ steps. That this is a $0$-budget strategy
follows from the invariant that after the first $k$ steps, there
is only one run that does not end in a challenge state, and in that
run, $v_{i}$ holds the value of $c_{i}$, and $m_{i}=-v_{i}$.

Conversely, assume that a winning strategy $\gamma$ exists. Then
the decision tree encoding the strategy has to be finite. In any such
strategy tree, challenge states may appear only at the leaves. Observe
that any challenge state $q_{c<0}$ or $q_{c\geq0}$ has a sibling
state $l\in L$. Furthermore, for all non-halting locations $l\neq\twocmhalt$,
$\funcapptrad{\nu}l=\bot$, and so no non-halting location can be
at the leaf of the strategy tree. Thus, some leaf of the strategy
tree has to be in a location $l=\twocmhalt$. Consider the finite
sequence of input symbols leading to this location. Because it is
a winning strategy, at each node of the tree, if the next input symbol
is $\tuple{s_{1},s_{2}}$, then $s_{1}$, $s_{2}$ are respectively
the signs of the values in $c_{1}$, $c_{2}$ after the machine executes
for the appropriate number of steps. Thus, this trace encodes a halting
run of the machine.\end{proof}

\section{\label{sec:Conclusion} Conclusion}

In this paper, we studied two decision problems for additive regular
functions: determining the register complexity, and alternating reachability
in $\acra$s. The register complexity is the largest number $k$ so
that every $\acra$ implementing $f$ has at least $k$ registers.
We developed an abstract characterization of register complexity as
separability and showed that computing it is $\pspacec$. We then
studied the reachability problem in alternating $\acra$s, and showed
that it is undecidable for $\dacra{\Z}$ and $\exptimec$ for $\dacra{\N}$
games. Future work includes proving similar characterizations and
providing algorithms for register minimization in more general models
such as streaming string transducers. String concatenation does not
form a commutative monoid, and the present paper is restricted to
unary operators (increment by constant), and so the technique does
not immediately carry over. Another interesting question is to find
a machine-independent characterization of regular functions $\func f{\kstar{\Sigma}}{\Z_{\bot}}$.
A third direction of work would be extending these ideas to trees
and studying their connection to alternating $\acra$s.

\bibliographystyle{plain}
\bibliography{canon-lncs-camera-Full}

\appendix

\end{document}